\documentclass[journal,twoside,web]{ieeecolor}

\usepackage[utf8]{inputenc}

\let\labelindent\relax
\usepackage{generic}
\usepackage{amssymb,amsmath,amsthm,mathrsfs,mathtools}
\usepackage[dvipsnames]{xcolor}
\usepackage{soul}
\usepackage{times}
\usepackage{graphicx}
\usepackage{dblfloatfix}
\usepackage{caption}
\usepackage{subcaption}
\usepackage[shortlabels]{enumitem}
\usepackage{makecell}

\usepackage{algpseudocode}
\usepackage{algorithm}
\usepackage{cite}

\newtheorem*{proposition*}{Proposition}
\newtheorem*{theorem*}{Theorem}
\newtheorem{assumption}{Assumption}
\newtheorem{lemma}{Lemma}
\newtheorem{theorem}{Theorem}
\newtheorem{definition}{Definition}
\newtheorem{problem}{Problem}
\newtheorem{remark}{Remark}

\newtheorem{proposition}{Proposition}

\usepackage{tikz}
\usetikzlibrary{arrows}
\usetikzlibrary{automata,arrows,positioning,calc}
\usepackage{circuitikz}
\usetikzlibrary{shapes}

\definecolor{babypink}{rgb}{0.96, 0.76, 0.76}
\definecolor{bananayellow}{rgb}{1.0, 0.88, 0.21}
\definecolor{amethyst}{rgb}{0.6, 0.4, 0.8}
\definecolor{blizzardblue}{rgb}{0.67, 0.9, 0.93}
\definecolor{aquamarine}{rgb}{0.5, 1.0, 0.83}
\definecolor{aureolin}{rgb}{0.99, 0.93, 0.0}
\definecolor{aqua}{rgb}{0.0, 1.0, 1.0}
\definecolor{caribbeangreen}{rgb}{0.0, 0.8, 0.6}
\definecolor{chartreuse(web)}{rgb}{0.5, 1.0, 0.0}
\definecolor{amber(sae/ece)}{rgb}{1.0, 0.49, 0.0}
\definecolor{apricot}{rgb}{0.98, 0.81, 0.69}
\definecolor{lightmauve}{rgb}{0.86, 0.82, 1.0}
\definecolor{lightsalmon}{rgb}{1.0, 0.63, 0.48}
\definecolor{electricblue}{rgb}{0.49, 0.98, 1.0}
\definecolor{gray(x11gray)}{rgb}{0.75, 0.75, 0.75}

\newcommand{\EE}{\mathbb{E}}

\newcommand{\Ns}{\mathcal{N}}

\newcommand{\Os}{\mathcal{O}}

\DeclarePairedDelimiter\ceil{\lceil}{\rceil}

\DeclareMathOperator*{\argmin}{argmin}

\definecolor{green1}{rgb}{0.2,0.7,0.2}

\title{\LARGE \bf Weighted Age of Information based Scheduling for Large Population Games on Networks}

\author{Shubham~Aggarwal, Muhammad~Aneeq~uz~Zaman, Melih Bastopcu, \textit{Member, IEEE}, and Tamer~Ba{\c s}ar, \textit{Life Fellow, IEEE}
\thanks{Research of the authors is supported in part by ARO MURI Grant AG285 and in part by AFOSR Grant FA9550-19-1-0353.}
\thanks{Shubham Aggarwal and Muhammad Aneeq uz Zaman are with the Coordinated Science Laboratory and the Department of Mechanical Science and Engineering at the University of Illinois at Urbana-Champaign (UIUC); Melih Bastopcu is with the Coordinated Science Laboratory at UIUC;  Tamer Ba{\c s}ar is with the Coordinated Science Laboratory and the Department of Electrical and Computer Engineering at UIUC. (Emails:
        {sa57@illinois.edu, mazaman2@illinois.edu, bastopcu@illinois.edu, basar1@illinois.edu)} 
}
}

\tikzset{ remember picture,
   switch/.style = {rectangle,
                    draw,align=center,
                    label={below:#1},
   },
}

\newsavebox\mybox
\savebox\mybox{%
\tikz\draw[line width=0.7pt] (-0.4,0)--(0,0)
								(0,0)--(0.4,0.4);%
}
\begin{document}
\tikzstyle{rect} = [draw,rectangle,fill = white!20,minimum width = 3pt, inner sep  = 5pt]
\tikzstyle{line} = [draw, -latex]
\tikzstyle{dline} = [draw, dash dot, -latex]

\maketitle
\thispagestyle{empty}
\begin{abstract}
In this paper, we consider a discrete-time multi-agent system involving $N$ cost-coupled networked rational agents solving a consensus problem and a central Base Station (BS), scheduling agent communications over a network. Due to a hard bandwidth constraint on the number of transmissions through the network, at most $R_d<N$ agents can concurrently access their state information through the network. Under standard assumptions on the information structure of the agents and the BS, we first show that the control actions of the agents are free of any \emph{dual effect}, allowing for separation between estimation and control problems at each agent. Next, we propose a \textit{weighted age of information (WAoI)} metric for the scheduling problem of the BS, where the weights depend on the estimation error of the agents. The BS aims to find the optimum scheduling policy that minimizes the WAoI, subject to the hard bandwidth constraint. Since this problem is NP-hard, we first relax the hard constraint to a soft \emph{update-rate} constraint, and then compute an optimal policy for the relaxed problem by reformulating it into a Markov Decision Process (MDP). This then inspires a sub-optimal policy for the bandwidth constrained problem, which is shown to approach the optimal policy as $N \rightarrow \infty$.
Next, we solve the consensus problem using the mean-field game framework wherein we first design decentralized control policies for a limiting case of the $N$--agent system (as $N \rightarrow \infty$). By explicitly constructing the mean-field system, we prove the existence and uniqueness of the mean-field equilibrium. Consequently, we show that the obtained equilibrium policies constitute an $\epsilon$--Nash equilibrium for the finite agent system. Finally, we validate the performance of both the scheduling and the control policies through numerical simulations.
\end{abstract}
\section{Introduction}
With the emergence of time-critical applications such as real-time monitoring in surveillance  systems, autonomous vehicular systems, internet-of-things and cyberphysical security,  \cite{shamma2008cooperative,zhang2021age,sandberg2015cyberphysical}, \textit{networked control systems} have garnered a lot of attention over the past decades, and promise interesting research directions within both the communication and the controls community. Such systems (often referred to as networked multi-agent systems) involve large populations of spatially distributed agents, in which both feedforward and feedback signals are exchanged between system components such as sensors, controllers and actuators through an (often) band-limited communication network \cite{zhang2015survey}. Such a scenario allows for remote information sharing and decentralized execution of the designated tasks. However, while decentralization reduces the storage complexity of servers at each system component, limited information availability directly affects the control performance of each agent. Thus, appropriate information structures \cite{malikopoulos2022team} need to be assigned to each system component alongside \emph{timely} and \emph{accurate} transmission of time-sensitive sensor measurements to the corresponding control units. All these factors are critical in determining the design of the optimal scheduling and control policies to maximize the system performance while meeting the system requirements, which forms the major objective of this paper. 

In this paper, we consider a discrete-time problem among $N+1$ players ($N$ agents and a Base Station) where the $N$ cost-coupled agents are solving a consensus problem, as in previous works \cite{uz2020approximate,aggarwal2022linear}. Each agent constitutes two active decision makers--a controller and an estimator--jointly minimizing a quadratic cost function in a team setting. The agents have access to their state information through a wireless communication network, controlled by the Base Station (BS), which acts as a central planner (or a global resource allocator \cite{molin2014optimal}). A schematic diagram depicting the hierarchical structure between the $N+1$ agents is shown in Fig.~\ref{Fig:Inf_flow}. As in the case of the real world wireless communication systems, the medium connecting the BS with the agents has a limited bandwidth of $R_d ~(< N)$ units. As a consequence of this communication bottleneck, the agents may have intermittent access to their state information. Hence, the BS must efficiently schedule transmissions so as to minimize the agents' estimation errors caused by the intermittent communication and also to enable timely communication at the control units. For this purpose, we propose a novel Age of Information (AoI) based metric, called the \emph{Weighted}-AoI (WAoI), to optimally schedule the wireless communication between the agents. This poses a novel $N+1$ player game where the $N$ agents are individually trying to solve a consensus problem among themselves while the BS is trying to minimize an aggregate estimation error based metric (to be precisely defined later) using a scheduling policy.
Due to the presence of a large population of agents communicating over the network, we solve the consensus problem using a mean-field game (MFG) setting. Further, the scheduling problem, where the goal is to minimize the WAoI subject to the limited bandwidth constraint at the scheduler, is an NP-hard combinatorics problem, and thus we first solve a relaxed problem using a Markov Decision Process (MDP) formulation. This then inspires an asymptotically optimal solution to the original scheduling problem as the agents grow in number.

\textbf{Related Literature: } Ever since the seminal works \cite{feldbaum1961dual,witsenhausen1968counterexample,witsenhausen1971separation}\footnote{See \cite{basar2008variations} for variations of the problem considered in \cite{witsenhausen1968counterexample}.} established no-dual effect and separation properties in systems involving partial observations/limited information, a lot of the literature has been concerned with the developments on information structures in stochastic decision making problems \cite{ho1978teams,mahajan2012information,malikopoulos2022team}. Since networked systems involve remote decision making, (un)availability of information at each decision maker has a huge impact on the control performance of the system. A number of works have considered networked control problems involving both estimation and control with uninterrupted \cite{bansal1989simultaneous,tatikonda2004stochastic}, and intermittent communications \cite{ramesh2013design,molin2012optimality,antunes2019consistent,imer2010optimal}. In \cite{ramesh2013design}, the authors consider both the state and innovations-based event-triggered schedulers to transmit information over a contention-based multiple access network, and derive conditions ensuring separation between the scheduler, observer and the controller. While \cite{molin2012optimality} studies the optimality of an event-based scheduler under constraints on the number of resource acquisitions, and on the average number of resource acquisitions, in \cite{antunes2019consistent}, the authors show the better performance of a consistent event-based scheduling policy over a periodic scheduling policy for partially observed discrete-time systems in a single-agent setting. Finally, in \cite{imer2010optimal}, the authors derive information disclosure policies between an observer and an estimator under a bandwidth constraint over the number of transmissions, in a team setting. Additionally, most earlier works \cite[Chapters 6-8]{molin2014optimal} exclude the set of communication instants from the information structure of the decision maker, the inclusion of which while it conveys side information on the current plant state  (even in the absence of communication between sensor and controller), makes the multi-agent problem, harder to solve. The inclusion of the same is considered in \cite[Chapter 2]{molin2014optimal}, in \cite{maity2020minimal} for a single agent setting, and in \cite{ramesh2013design} for a multi-agent setting, but the considered overall design is only suboptimal. What also differentiates our setting from the ones in \cite{maity2020minimal,soleymani2017value} is that they consider local schedulers (or event-triggers) within each feedback loop while we consider continuous sensing (at each discrete instant) and a hierarchical structure between agents and the BS as in Fig. \ref{Fig:Inf_flow}. Further, the above-mentioned works contain no strategic interaction between the agents as in a game setting, and an additional challenge with the above designs is that of scalability to large population systems.

To appropriately handle the concerns posed by increasing number of agents in the population, there has been a rapid growth in interest in the MFG setting, since the early works \cite{huang2006large,lasry2007mean,huang2007large}, to solve multi-agent problems with strategic interactions. This is due to its efficiency in handling the scalability issue that emerges in finite population games as the number of agents increases. The key idea in MFGs is that, as the number of agents approaches infinity, the agents become indistinguishable and the effect of individual deviations on equilibrium disappears. This leads to an aggregation effect and the game problem consequently reduces to a stochastic optimal control problem of a generic agent alongside a consistency condition. While the above works consider continuous-time agent dynamics, the discrete-time counterpart is considered in \cite{uz2020approximate,uz2020reinforcement} and \cite{moon2014discrete}, among others. The LQ-MFG setting \cite{bensoussan2016linear,huang2018linear} with linear agent dynamics and quadratic cost functions (with consensus terms) serves as a standard, but significant benchmark for general MFGs. Most works in LQ-MFGs consider continuous and reliable communication, with the possible exception of \cite{moon2014discrete}, which considers passive (probabilistic) scheduling over unreliable channels. Under suitable assumptions on the probabilities of erasure of actuation and measurement signals, the paper derives approximate-Nash equilibrium strategies for the finite agent game. A recent work \cite{aggarwal2022linear} deals with a game where the agents are cost-coupled and have access to their state through a noisy intermittent (individual) channel. The problem in this paper, however, differs from the ones in \cite{moon2014discrete,aggarwal2022linear} since we \textit{actively} schedule transmissions over the transmission medium, and extend the unconstrained setting to the case where the agents are connected via a wireless-network structure with a limited bandwidth. We utilize the emerging notion of AoI to devise optimal scheduling policies over the network.

Age of information has been introduced in \cite{Kaul12a} to measure the timeliness of information in communication networks. After the initial works on AoI which mostly focus on the queueing aspect, scheduling problems in AoI have received significant attention recently. The maximum age first (MAF) policy has been shown to minimize the average AoI for a multi-source system where only one sample can be transmitted at a time with random delay \cite{Bedewy19} and over a channel with transmission errors \cite{Kadota18a}. In \cite{Beytur18a}, the authors consider Maximum-age-difference (MAD) policy for preemptive and non-preemptive scheduling policies and numerically observed its performance. In \cite{Zhong19}, it is shown that the Maximum Weighted Age Reduction (MWAR) policy minimizes a special class of age-penalty functions. The optimality of Whittle’s index policy to minimize AoI has been studied in \cite{maatouk2020optimality}. Age-optimal scheduling policies are obtained by using MDP formulation in \cite{tang2020scheduling, chen2021minimizing, Yun18, Hsu2020,hatami2022demand}.

Recently, AoI has also been studied in the context of networked control problems. In \cite{Nadeem2022}, the mean-square estimation (MSE) performance of a system is optimized through HARQ schemes by using MDP formulation. In \cite{ayan2019age}, the authors compare the performances of AoI based and Value of Information (VoI) based scheduling algorithms where VoI is measured through uncertainty reduction upon information delivery. While the paper \cite{Ayan2020} proposes a discounted error scheduler by using MDP approach for a truncated state-space of the AoI,
\cite{Ayan2021} studies a resource allocation problem to minimize MSE through AoI. Further, we note also that in contrast to \cite{Ayan2020}, we do not truncate the state space for deriving scheduling policies and the results in our work apply to an unbounded state space, whereby we show by analysis that the optimal policy is characterized as a threshold policy. Finally, \cite{wang2021value} has proposed a VoI based event-triggering policy which achieves a lower MSE compared to the AoI-based and periodic updating policies. Surveys \cite{SunSurvey, Yates20a} provide a more detailed literature review on AoI. Finally, we note that \cite{zhang2021age} considers an AoI motivated internet-of-things problem within the realm of MFGs and is formulated in the spirit of \cite{parise2014mean,grammatico2015decentralized}; however, it is fundamentally different from our problem since our area of focus is scheduling policy design from the scheduler to the individual systems and not on the broadcast of information over the uplink.


\textbf{Contribution:} The main contributions of this work and its comparison with available literature are as follows.
\begin{itemize}
    \item We solve a large population game problem involving $N$ linear networked dynamical agents interacting strategically with each other to minimize individual but coupled quadratic cost functions (as in \cite{moon2014discrete,aggarwal2022linear}) and a base station (BS) aiming to minimize a performance measure by scheduling communication over a hard-bandwidth constrained wireless medium.
    \item Inspired by the age of incorrect information (AoII) metric \cite{maatouk2020age1}, we introduce a weighted AoI metric as the performance measure for the BS, which is a function of the estimation error between the plant state and the controller state (which results due to intermittent information transmission by the BS to the controller) and the AoI at the controller. Such a metric is in contrast to those available in the literature, which either include only the current AoI as the cost \cite{kadota2018scheduling,tang2020scheduling,Bedewy19} to transmit information in a timely fashion or only the VoI to minimize the system uncertainty \cite{ayan2019age,wang2021value}. Additionally, metrics based on AoII are control unaware\footnote{This term is discussed in detail in the next section.} and do not involve the presence of a feedback signal. Thus, our metric is novel and appropriately motivated.
    \item We allow the agents to keep track of the set of scheduling instants as prescribed by the BS. With its incorporation into the information structure, we establish separation between the controller and the decoder by treating it as a special case of \cite{ramesh2013design}. Such an information structure is inspired from \cite{maity2020minimal}, where the authors, however, consider a single-agent problem. Further, it is in contrast to similar multi-agent problems \cite[Chapters 6-8]{molin2014optimal}, where such information is discarded.
    \item Since the optimal scheduling problem of the BS belongs to the class of multi-armed restless bandits problems \cite{maatouk2020optimality}, finding an optimal policy is not possible. Hence, we propose a suboptimal solution to the scheduling problem. We show that our proposed solution is asymptotically optimal in the limit as $N \rightarrow \infty$. This is in contrast to results in the literature which either consider concave or affine running costs for the scheduler \cite{chen2021optimizing,tang2020scheduling}. Further, the asymptotic results proposed in \cite{hatami2022demand} are oblivious to the rate of convergence in the stand-alone scheduling problem. On the contrary, in this work, the optimization problem posed at the scheduler is (naturally) non-convex and the proposed solution admits a  rate of convergence of the order of $\mathcal{O}\left( \frac{1}{\sqrt{N}}\right)$ for the AoI evolving in an unbounded state space.
    \item Finally, due to the difficulty of dealing with the consensus problem in a \emph{large} (but not infinite) population setting, we use the mean-field game framework to construct $\epsilon$-Nash strategies for the $N$--agent game problem. 
    Using the scheduling policy constructed above, we prove the existence of a unique mean-field equilibrium (MFE) and also establish its linearity. Furthermore, we prove that $\epsilon = \mathcal{O}\left(\frac{1}{\sqrt{\min_{\theta \in \Theta}|N_{\theta}|}}\right)$, where $|N_{\theta}|$ denotes the number of agents of type $\theta \in \Theta$, which is stronger than the asymptotic result in \cite{moon2014discrete}.
\end{itemize}

\textbf{Organization:} In Section \ref{sec:Prob_Form}, we formulate the $N+1$ player game, where the agents aim to achieve consensus and the BS aims to minimize the average WAoI. In Section \ref{sec:CSP}, we analyze the centralized scheduling problem of the BS and transform it into $N$ decentralized scheduling problems by relaxing the hard bandwidth limit. The hard bandwidth limit is first relaxed to a soft rate limit, and by using an MDP formulation, we are able to show that a randomized threshold policy is optimal for the relaxed problem. This policy then inspires an approximately optimal scheduling policy which is shown to be asymptotically optimal for the original bandwidth-limited scheduling problem (Section \ref{sec:DSP}). In Section \ref{sec:cons_prob}, we solve the consensus problem between agents by first considering an MFG (constituting infinitely many agents) and proving the existence and uniqueness of the mean-field equilibrium. Then, we show that the decentralized control policies obtained by using the MF analysis constitute an $\epsilon$--Nash equilibrium for the finite agent game problem. In Section \ref{sec:sims}, by using numerical analysis we demonstrate the empirical performance of both the control as well as the scheduling policies. The paper is concluded in Section \ref{sec:conclusion} with some major highlights.



\begin{figure}[t]
	\centerline{\includegraphics[width=1\columnwidth]{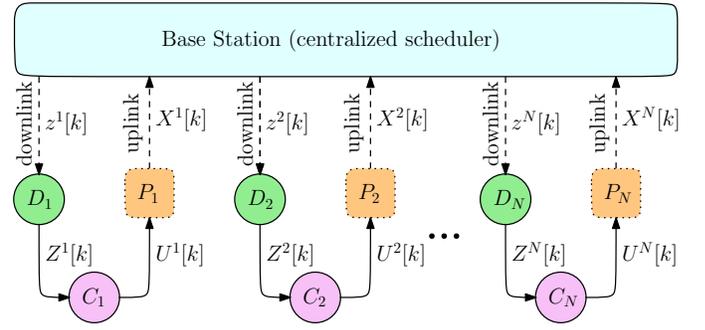}}
	\caption{\small{Schematic diagram of the networked system showing hierarchical levels of decision making at the Base Station and the agents. Solid blocks (Base Station, Decoders, Controllers) denote active decision makers while dotted blocks (plants) denote passive components. Dashed lines denote a wireless transfer medium while bold lines denote wired transfer. The decoder and the controller of each agent are within a team setting minimizing a quadratic objective while there is a game being played between $N$ agents. The BS on the top acts as a global decision maker minimizing its own cost function.}}
	\label{Fig:Inf_flow}
    \vspace{-0.4cm} 
\end{figure}

\textbf{Notations:} $\mathbb{Z}^+$ denotes the set of non-negative integers and $\mathbb{N} = \mathbb{Z}^+\setminus \{0\}$. For a matrix $S$ and vector $x$, $\|x\|^2_S:= x^\top S x$. Further, $[N]:= \{1,2, \cdots, N\}$ and $tr(\cdot)$ denotes the trace of its argument. The symbol $\Gamma^d := \Gamma^{d,1} \times \cdots \times \Gamma^{d,N}$ denotes the set of scheduling policies. While Euclidean norm for vectors or the induced 2-norm for matrices is denoted by $\|\cdot\|$, $\|\cdot\|_F$ denotes the Frobenius norm of its argument. For real functions, $\operatorname{A}(\alpha)=\mathcal{O}(\operatorname{B}(\alpha))$ is equivalent to the statement that there exists $\operatorname{K} >0$ such that $\lim_{\alpha \rightarrow \infty}\frac{|\operatorname{A}(\alpha)|}{|\operatorname{B}(\alpha)|} = K$. 
The expression $Y_{a:a+k}:=\{Y[a], \cdots, Y[a+k]\}, ~k \geq 0$. All empty summations are set to zero, for example, $\sum_{l=1}^0(\cdot) = 0$. For any two matrices $X,Y \geq 0$, $X \geq Y$ means $X-Y \geq 0$.

\section{System Model and Problem Formulation}
\label{sec:Prob_Form}
We start by formulating the $N+1$ player game, which can be expressed as 1) a consensus problem between $N$ agents and 2) a centralized scheduling problem of the BS.
\subsection{Consensus Problem}\label{subsec:cons_prob}
We consider a discrete-time $N$-agent game on an infinite horizon, communicating over a wireless network. 
Each agent $i$ consists of the sub-component tuple $(P_i,D_i,C_i)$, i.e., a plant, a decoder and a controller, as shown in Fig. \ref{Fig:Inf_flow}. The plant dynamics evolve according to a linear difference equation
\begin{align}\label{system}
\!\!\!X^i[k\!+\!1] \!=\! A(\theta_i)X^i[k] \!+\! B(\theta_i)U^i[k] \!+\! W^i[k],
\end{align}
for timestep $k \in \mathbb{Z}^+$ and agent $i \!\in\! [N]$. Here $X^i[k] \in \mathbb{R}^n$ and $U^i[k] \in \mathbb{R}^m$ are the state and control actions, respectively of agent $i$. The process noise for agent $i$, $ W^i[k] \in \mathbb{R}^n$ has zero mean and  bounded positive-definite covariance $K_W(\theta_i)$. Further, we assume that $\sup_{1 \leq i \leq p}K_W(\theta_i) \leq K_0$. Agent $i$'s initial state $X^i[0]$ is independent of the process noise and is assumed to have a symmetric density with mean $\nu_{\theta_i,0}$ and bounded positive-definite covariance matrix $\Sigma_x$. The time-invariant system matrices $A(\theta_i) \in \mathbb{R}^{n \times n}$ and $B(\theta_i) \in \mathbb{R}^{n \times m}$ depend on the agent type $\theta_i \in \Theta := \{\theta_1, \cdots, \theta_p\}$ which is chosen according to the empirical probability mass function $\mathbb{P}_N(\theta =\theta_i), ~~\theta \in \Theta$.
It is further assumed that $ \lvert \mathbb{P}_N(\theta) - \mathbb{P}(\theta) \rvert = \Os(1/N)$ for all $\theta \in \Theta$, where $\mathbb{P}(\theta)$ is the limiting distribution.

Next, for each $P_i$, its full-state information is relayed to the decoder $D_i$ through a two-hop network (called uplink and downlink) via a centralized BS, which is then communicated to the controller $C_i$ for generating an actuation signal. The uplink capacity of the wireless network is $R_{\text{u}} = N$ and the downlink capacity is $R_{d} < N$. This implies that only the downlink acts as a bottleneck for transmission of information from the plants to their respective controllers.
We now state the following assumption on information transmission over the network.
\begin{assumption}\label{As1}
We assume that,
\begin{enumerate}[(i)]
\item the wireless links connecting system components are perfect (free of any noise and transmission errors), and
\item the BS can send measurements to the corresponding controllers instantaneously. 
\end{enumerate}
\end{assumption}
As a result of Assumption \ref{As1}, the information can be transmitted from the plant to the controller for its next action without any delay, if the BS decides to send an update for that particular agent. Further, Assumption \ref{As1}(i) can be relaxed to include for instance, additive noise channels or erasure channels. This more general model would require additional back-channels from the decoders/estimators to the scheduler to communicate the best estimate of the state in order to entail no-dual effect as in  \cite{aggarwal2022linear}. This no-dual effect property will be instrumental in the design of both scheduling and control policies as shown later. Next, we note that in many information updating systems, the required transmission time for the update is usually less than the time unit used to measure the dynamics of a process \cite{chen2021minimizing}. Hence, Assumption~\ref{As1}(ii) is justified.\footnote{If we relax this assumption and allow delays in the update transmissions, one needs to keep track of the additional delay in the AoI of each agent, and we leave such a direction as future work. A preliminary result considering additional delays can be found in \cite{wang2021value}.}

 Then, under Assumption \ref{As1}, the state of the $i^{th}$ plant as observed by the $i^{th}$ decoder is given as:
\begin{align}\label{decoder_signal}
z^i[k] = \begin{cases}
X^i[k], & \text{if} ~~\zeta^{d,i}[k] = 1, \\
\varphi, & \text{if} ~~\zeta^{d,i}[k] = 0,
\end{cases}
\end{align}
where $\zeta^{d,i}[k] = 1$ denotes the scheduling decision such that \emph{current state} information is transmitted to the $i^{th}$ decoder (over the downlink) while $\zeta^{d,i}[k] = 0 $ stands for no transmission (or $\varphi$). 
Additionally, between transmission times, the decoder calculates the best minimum mean square estimate (MMSE), $Z^i[k]$, based on its information history $I^{d,i}[k]:=\left\lbrace z^i_{0:k},\zeta^{d,i}_{0:k},U^i_{0:k-1},Z^i_{0:k-1}\right\rbrace$, which equals the conditional expectation of the state given the information structure ($\mathbb{E}\left[X^i[k] \mid I^{d,i}[k] \right]$). The same is then sent to the controller for which an explicit expression is derived later. We adopt the convention that $z^i[-1],Z^i[-1],U^i[-1]$ are all 0 and $W^i[-1]  = X^i[0]- Z^i[0]$ for all $i$. Typically, in the game problems as formulated above, the control action of the agent $i$ can depend on other agents' state and control actions, and hence, the information history of the $i^{th}$ controller would be given by $I^{c,con,i}[k]:= \{U^i_{0:k-1},Z^i_{0:k}\}_{i \in [N]}$. Here, $I^{c,con,i}[k]$ denotes a centralized information structure whereby the controller has the knowledge of not only its own but also of other agents' controller states and actions. This entails that $\mathcal{M}^{c,con}_i = \{\pi_c^i \mid \pi_c^i ~\text{is adapted to } \sigma(I^{c,con,i}[s], ~s = 0,1, \cdots, k)\}$, where $\sigma(\cdot)$ is the $\sigma$--algebra adapted to its argument and $\mathcal{M}^{c,con}_i$ is the space of admissible centralized control policies for agent $i$ \cite{aggarwal2022linear}. We also note that we do not require a back-channel from the controller to the decoder since the control policy (designed later in Section \ref{sec:cons_prob}) can be input \textit{a priori} into the decoder; and consequently, it can compute the control actions at its own end.

Now, each agent $i$ aims to minimize its infinite-horizon average cost function
\begin{align}\label{LQT}
&J_i^N(\pi_c^i,\pi_c^{-i}) :=  \limsup_{T\rightarrow \infty} \\&  \frac{1}{T}\mathbb{E}\left[ \sum_{k=0}^{T-1}\|X^i[k]-\frac{1}{N}\sum_{j=1}^{N}{X^j[k]}\|^2_{Q(\theta_i)} + \|U^i[k]\|^2_{R(\theta_i)}\right],  \nonumber
\end{align}
where $Q(\theta_i) \geq 0$ and $R(\theta_i)>0$. The coupling between agents enters through the consensus-like term $\frac{1}{N}\sum_{j=1}^{N}{X^j[k]}$. The cost function penalizes deviations from the consensus term and large control effort. 
We define $\pi_c^{-i}:= (\pi_c^1,\cdots, \pi_c^{i-1}, \pi_c^{i+1}, \cdots, \pi_c^N)$, where $\pi_c^i := (\pi_c^i[1],~\pi_c^i[2], \cdots, ) \in \mathcal{M}^{c,con}_i$ denotes a control policy for the $i^{th}$ agent. 
Finally, the expectation in \eqref{LQT} is taken with respect to the noise statistics and the initial state distribution.
We note now that having access to (and keeping track of) the information of other agents in a large population setting is quite difficult, and hence, we will resort to the MFG framework to characterize decentralized control policies whereby decisions are made based on an agent's local information. This is defined and treated more formally in Section \ref{sec:cons_prob}. We also remark here that the estimator and the controller \textit{work together} in a team setting to minimize \eqref{LQT}, and as we will see later, can be designed independently of each other to lead to a globally optimal design. Finally, the BS acts as a centralized scheduler and schedules transmissions over the downlink using an optimal scheduling policy in order to minimize the WAoI across all agents, as we define below.
\subsection{Centralized Scheduling Problem}
Consider the most recently received observations by the controller $i$ as $z^i[s^i[k]]$, where $s^i[k] = \sup \{s \in \mathbb{Z}^+: s \le k, z^i[s] \neq \varphi\}$ denotes the latest transmission time. By definition, the AoI is the time elapsed since the generation time-stamp of the most recent packet at the plant. Thus, the AoI $\Delta^i[k]$ at the controller $C_i$ is given as
$\Delta^i[k] = k - s^i[k]$. 
More precisely, we have
\begin{align}
\Delta^i[k+1] = \begin{cases}
0, & \text{if} ~~\zeta^{d,i}[k] = 1, \\
\Delta^i[k] + 1, & \text{if} ~~\zeta^{d,i}[k] = 0.
\end{cases}
\end{align}
Thus, the scheduling problem with hard-bandwidth constraint to be solved by the BS is:
\begin{problem}[Control-Aware Constrained AoI minimization]\label{Problem1}
\begin{align*}
\inf_{\gamma^{d} \in \Gamma^{d}} ~ &  J^S (\gamma^{d}):= \limsup_{T \rightarrow \infty}\frac{1}{T} \mathbb{E}\Bigg[\frac{1}{N}\sum\limits_{k=0}^{T-1}\sum\limits_{i=1}^{N}\eta^{i}[k]\Delta^i[k]\Bigg]  \\
\mbox{s.t.} \quad & \sum_{i=1}^N \zeta^{d,i}[k]  \le R_{d}, ~\forall k,
\end{align*}
where 
${\gamma^{d}} = [\gamma^{d,1}, \cdots, \gamma^{d,N}]^\top$, $\Gamma^{d}$ is the space of scheduling policies across all agents, and $\zeta^{d,i}[k]$ is the scheduling decision chosen from an admissible policy $\gamma^{d,i}$, for all $i$. Further, we define $\eta^i[k]$ as the importance weight associated with agent $i$ given by $\eta^i[k]:= \mathbb{E}[\|e^i[k]\|^2]$, where $e^i[k] := X^i[k]-Z^i[k]$ is the estimation error between the plant state and the controller state at instant $k$. Finally, the expectation is taken over the stochasticity induced by random policies (which, as we will see later provides an optimal solution to the above problem). 
\end{problem}
\section{Centralized Scheduling Problem Analysis}
\label{sec:CSP}
We start this section by computing the best estimate at the decoder. Then, we show that since the system matrices are time-invariant, we can transform the importance weights (which are functions of the errors) into equivalent functions of the AoI for the respective plants. Also, to avoid cluttering notations, we use the shorthands $A_i:= A(\theta_i)$, $B_i:= B(\theta_i)$, and $K_{W^i}:= K_W(\theta_i)$, unless specified otherwise.

Based on the signal input to the decoder \eqref{decoder_signal}, the decoder computes the best estimate of the state as
\begin{align*}
    Z^i[k] = \begin{cases}
X^i[k], & \text{if} ~~\zeta^{d,i}[k] = 1, \\
\mathbb{E}\left[X^i[k] \mid I^{d,i}[k], \zeta^{d,i}[k] = 0 \right], & \text{if} ~~\zeta^{d,i}[k] = 0,
\end{cases}
\end{align*}
where
\begin{align}
    & \mathbb{E}\left[X^i[k] \mid I^{d,i}[k], \zeta^{d,i}[k] = 0 \right] =\nonumber \\
    &  A_i\mathbb{E}\left[X^i[k-1]\!\mid \!I^{d,i}[k-1]\right] \!+ \!B_iU^i[k-1] \!+\! \hat{W}^i[k-1],
\end{align}
where $\hat{W}^i[k-1]:=\mathbb{E}\left[ W^i[k-1] \mid  \zeta^{d,i}[k] = 0 \right]$. This yields
\begin{align*}
    & \!Z^i[k] \!= \!\!\begin{cases}
\!X^i[k], & \!\!\!\!\!\text{if} ~\zeta^{d,i}[k] \!= \!1, \\
\!A_iZ^i[k-1]\! + \!B_iU^i[k-1]\! +\!\hat{W}^i[k-1], & \!\!\!\!\!\text{if} ~\zeta^{d,i}[k] \!= \!0.
\end{cases}
\end{align*}
Note that the presence of transmission instants in the conditioning leads to the extra term $\hat{W}^i[k-1]$. While this conveys additional information on the state of the plant in the absence of any communication between the BS and the decoder \cite{maity2020minimal}, it is typically hard to compute optimally. However, we show now that this term equals 0 due to the setup in Fig. \ref{Fig:Inf_flow}. 

\begin{figure}[t]
	\centerline{\includegraphics[width=0.80\columnwidth]{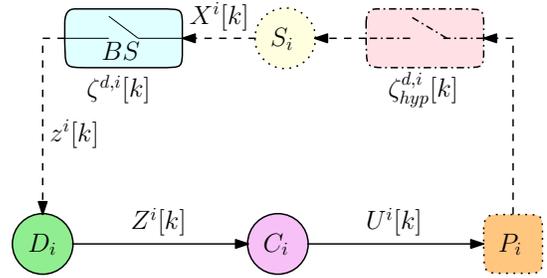}}
	\caption{\small{Focused view of a single agent in Fig. \ref{Fig:Inf_flow}. Solid blocks (Base Station $BS$, Decoder $D_i$, Controller $C_i$) denote active decision makers while dotted block (plant $P_i$) denotes passive component. Dashed lines denote a wireless transfer medium while bold lines denote wired transfer.  Additionally, a \textit{hypothetical local scheduler} (dash-dotted and pink colored box) is inserted between the plant and the BS which connects both if $\zeta_{hyp}^{d,i}[k] = 1$.
	}}
	\label{Fig:magnified}
    \vspace{-0.4cm} 
\end{figure}

Consider the setup shown in Fig.~\ref{Fig:magnified}, where we focus on a single agent and additionally add a hypothetical local scheduler (dash-dotted) within the feedback loop, generating an output of $\zeta_{hyp}^{d,i}[k] \equiv 1, \forall k \geq 0$. Thus, we can define $\zeta_{hyp}^{d,i}[k]:= f(e^i[k])$, which is clearly an even mapping of the argument (i.e., $f(r) = f(-r)$). It then follows from \cite{ramesh2013design} and assumptions of symmetric densities of initial state and the noise process that $\hat{W}^i[k-1] = 0$. Hence, the optimal decoder is given by
\begin{align}\label{Recursive}
Z^i[k] = \begin{cases}
X^i[k], & \text{if} ~~\zeta^{d,i}[k] = 1, \\
A_iZ^i[k-1] \!+\! B_iU^i[k-1], & \text{if} ~~\zeta^{d,i}[k] = 0,
\end{cases}
\end{align}
 where $A_iZ^i[k-1] + B_iU^i[k-1] $ can be thought of as the recursive estimate calculated between transmission instants based on the information history. The signal $Z^i[k]$ can now be computed easily using \eqref{Recursive}. We summarize the above results in the following lemma, which entails a closed form of expression of the WAoI as a function of the current AoI of the system.

\begin{lemma}\label{L1}
The estimation error $e^i[k]$ for all agents is independent of the control inputs, and hence there is no dual effect of control \cite{aggarwal2022linear}. Moreover,
\begin{align}\label{error_definition}
e^i[k] = \begin{cases}
0, & \text{if} ~~\zeta^{d,i}[k] = 1, \\
\sum_{l=1}^{\Delta^i[k]}{A_i^{l-1}W^i[k-l]}, & \text{if} ~~\zeta^{d,i}[k] = 0,
\end{cases}
\end{align}
and the covariance of the estimation error can be formulated in terms of the AoI, i.e., $\mathbb{E}[\|e^i[k]\|^2] = h(\Delta^i[k],A_i,K_{W^i})$, which is given in \eqref{temp1}.
\end{lemma}
\begin{proof}
We note that the decoder output at a non-transmission instant $k$, can be rewritten by using \eqref{Recursive} as:
\begin{align*}
Z^i[k] = A_i^{\Delta^i[k]}z^i[s^i[k]] + \sum_{l=1}^{\Delta^i[k]}A_i^{l-1}B_iU^i[k-l].
\end{align*}

The result then follows immediately from the definition of $e^i[k]$. In addition, by using \eqref{error_definition}, the covariance of the estimation error is given by
\begin{align}\label{temp1}
h(\Delta^i[k],A_i,K_{W^i}) = \sum_{l=1}^{\Delta^i[k]}{tr\left({{A_i^{l-1}}^\top{A_i^{l-1}}K_{W^i}}\right)},
\end{align}
which completes the proof.
\end{proof}

Thus, Problem \ref{Problem1} can be equivalently written as:

\begin{problem}[Control-Aware Constrained AoI minimization]\label{Problem2}
\begin{align}
 \inf_{\gamma^{d} \in \Gamma^{d}} ~ & {J}^S(\gamma^{d}) = \nonumber \\ 
&\hspace{0.1cm} \limsup_{T \rightarrow \infty} \frac{1}{T} \mathbb{E}\Bigg[\frac{1}{N}\sum_{k=0}^{T-1}\sum_{i=1}^{N}\underbrace{h(\Delta^i[k],A_i,K_{W^i})\Delta^i[k]}_{=: g(\Delta^i[k],A_i,K_{W^i})}\Bigg] \nonumber\\
\mbox{s.t.} \quad & \sum_{i=1}^N \zeta^{d,i}[k]  \le R_{d}, ~\forall k.\label{AoI-constraint1}
\end{align}
\end{problem}
We note that Problem \ref{Problem2} does not only depend on the AoI but also on the system parameters of each agent such as $A_i$ and $K_{W^i}$. This entails that $g(\Delta^i[k],A_i,K_{W^i})$, which is nonlinear in the AoI of an agent, takes into account the system dynamics as well. Hence, we refer to it as \emph{Control-Aware AoI} \cite{ayan2019age}. Additionally, as a consequence of this reformulation, the scheduler does not have to store the plant states and the controller actions for any agent. Thus, no additional storage space is required for the scheduler.

\begin{remark}
Note that at this point, it may be tempting to remove the AoI term, $\Delta^i[k]$, from multiplication in the cost function since the error metric $h(\Delta^i[k],A_i,K_{W^i})$ is already a function of it. However, we note that if this term is removed, then, in general, the error over an infinite horizon may approach a finite limit (for instance, consider stable agents). This in turn may cause the AoI of those agents to approach infinity since a trigger for information transmission may never be generated for these agents (as will also be clear from Lemma \ref{L3} and the proof of Theorem \ref{Theorem_det} later). This phenomenon is additionally seen from the empirical analysis in \cite{ayan2019age}. Thus, the use of a weighted metric as in Problem \ref{Problem2}  penalizes both the error as well as the AoI and is appropriately justified.
\end{remark}
In the sequel, we first solve the scheduling problem in Section~\ref{sec:DSP} and then solve the consensus problem in Section \ref{sec:cons_prob}. We start by observing that Problem \ref{Problem2} is an intractable combinatorics problem due to the presence of the bandwidth constraint \eqref{AoI-constraint1}, and is difficult to handle. Thus, we first present a relaxed optimization problem involving a time-averaged soft rate constraint on the frequency of transmissions. This entails that more than $R_{d}$ users are able to transmit over the channel at some times as long as the average transmissions satisfy the average constraint over the infinite horizon. Consequently, we bring the $N$--agent scheduling problem into a single agent discrete-time MDP to find the optimal scheduling strategies.
Next, we define a relaxed version of Problem \ref{Problem2} as follows:
\begin{problem}[Relaxed Control-Aware Constrained AoI minimization]\label{Problem3}
\begin{align}
\inf_{\gamma^{d} \in \Gamma^{d}}  \quad & \limsup_{T \rightarrow \infty}\frac{1}{T} \mathbb{E}\Bigg[\frac{1}{N}\sum_{k=0}^{T-1}\sum_{i=1}^{N}g(\Delta^i[k],A_i,K_{W^i})\Bigg] \nonumber \\
\mbox{s.t.} \quad & \limsup_{T \rightarrow \infty} \frac{1}{T}\sum_{k=0}^T\sum_{i=1}^N \zeta^{d,i}[k]  \le R_{d}.\label{AoI-constraint2}
\end{align}
\end{problem}
Note that the above problem consists of a time-averaged soft rate constraint \eqref{AoI-constraint2} over the allowed transmissions rather than a bandwidth constraint as in \eqref{AoI-constraint1}. For this problem, we introduce the Lagrangian function as:
\begin{align*}
& \mathcal{L}(\gamma^{d},\lambda)  :=  \limsup_{T \rightarrow \infty}\frac{1}{T} \mathbb{E}\Bigg[\frac{1}{N}\sum_{k=0}^{T-1}\sum_{i=1}^{N}\bigg[g(\Delta^i[k],A_i,K_{W^i})   \\ &  \hspace{1.5cm}  +\lambda \left(\zeta^{d,i}[k] - \frac{R_{d}}{N}\right)\bigg]\Bigg],
\end{align*}
where $\lambda \ge 0$ is the Lagrange multiplier \cite{Boyd04}. Such a multiplier can be thought of as the cost of scheduling for each agent over the channel. Thus, for a fixed $\lambda$, the decoupled single user optimization problem is defined as follows:
\begin{problem}[Decoupled Control-Aware Unconstrained AoI minimization]\label{Problem4}
\begin{align}
\!\inf_{\gamma^{d,i} \!\in \Gamma^{d,i}} \!\limsup_{T \rightarrow \infty} \!\frac{1}{T}\mathbb{E}\Bigg[\!\sum_{k=0}^{T\!-\!1} \!g(\Delta^i[k],A_i,K_{W^i}\!)\!+\! \lambda\zeta^{d,i}[k]\Bigg]\nonumber .
\end{align}
\end{problem}
Since Problem \ref{Problem4} is solved for a single user, we henceforth drop the superscript $i$ until mentioned otherwise.

\section{Decentralized Scheduling Problem}
\label{sec:DSP}
Next, we formulate Problem \ref{Problem4} into a 
discrete-time controlled MDP $\operatorname{M}$, defined by the quadruplet $\operatorname{M}:= (\mathcal{S}, \mathcal{A}, \mathcal{P}, C)$. The \textit{state space} $\mathcal{S} = \mathbb{Z}^+$ is the set of all possible AoI of the agent 
and is countably infinite. The \textit{action space} is $\mathcal{A} = \{0,1\}$, 
where $a=1$ denotes that the agent is connected over the channel while $a=0$ stands for no transmission. Note here that $a$ is different from $\zeta^{d}[k]$, which is constrained by the hard-bandwidth limit. The \textit{probability transition function} $\mathcal{P}$ denotes the evolution of the states of the controlled system. When $a =0$, we have $\mathcal{P}(\Delta \rightarrow \Delta + 1) = 1$. When $ a=1$, we have $\mathcal{P}(\Delta \rightarrow 0) = 1$. We further note that although the states evolve deterministically, 
writing them in the form of an MDP will 
simplify the notation. The \textit{one-stage cost} $C(\Delta, a) = g(\Delta,A,K_{W}) + \lambda a$ denotes the cost incurred when an action $a$ is taken at the state $\Delta$.
Next, we formally define the decentralized scheduling problem.
\begin{problem}[Decentralized Scheduling Problem]\label{Problem5}
\begin{align}\label{cost-to-go}
\hspace{-0.1cm} \inf_{\gamma^{d} \in \Gamma^{d}} \!V(\Delta,\gamma^d) \!:=\!\left\lbrace\!\limsup_{T \rightarrow \infty} \frac{1}{T}\mathbb{E}\!\Bigg[\sum_{k=0}^{T-1} C(\Delta[k], a[k])\Bigg]\right\rbrace.
\end{align}
%
\end{problem}
%
In the following, we provide the solution to Problem~\ref{Problem5}.
\subsection{Single-agent Deterministic Scheduling Policy}
We first prove the following lemma, which is used in the main theorem characterizing optimality of a threshold policy.
\begin{lemma}\label{L3}
$g(\Delta[k],A,K_{W})$ has the following properties:
\begin{enumerate}
\item It is an increasing function of $\Delta[k]$ between any two successive transmission instants $k_1$ and $k_2$ such that $k_1 \leq k \leq k_2$.
\item $\lim\limits_{\Delta[k] \rightarrow \infty} g(\Delta[k],A,K_{W}) =\infty$.
\end{enumerate}
\end{lemma}
\begin{proof}
Note that $g(\Delta[k],A,K_{W}) = h(\Delta[k],A,K_{W}) \Delta[k]$. Since $tr(\cdot)$ is the sum of eigenvalues of the argument matrix, we see that the quantity $tr\left({{A^{l-1}}^\top{A^{l-1}}K_{W}}\right)$ is always non-negative since the argument matrix product has non-negative eigenvalues. Thus, for any $\Delta[k_1] \le \Delta[k_2]$, we have $h(\Delta[k_1],A,K_{W}) \le h(\Delta[k_2],A,K_{W})$. Since $g(\cdot,\cdot,\cdot) =h(\cdot,\cdot,\cdot)\Delta[k],$ we get both 1) and 2).
\end{proof}

\begin{definition}
A scheduling policy $\gamma^d$ for the MDP $\operatorname{M}$ is $g(\Delta)$--optimal if it infimizes the time-average cost $V(\Delta,\gamma^d).$ Moreover, we define the optimal cost as $\sigma^* = \inf_{\gamma^d \in \Gamma^d} V(\Delta,\gamma^d).$
\end{definition}

We first recall that the expectation in \eqref{cost-to-go} is only due to the stochastic nature of the scheduling policy. However, in the next two subsections, we characterize a deterministic optimal solution to Problem \ref{Problem5}. Thus, we remove the expectation operator in \eqref{cost-to-go} and proceed with further analysis. First, we have the following theorem.

{\begin{theorem}\label{Theorem_det}
There exists a $g(\Delta)$--optimal stationary deterministic policy solving the Decentralized Scheduling Problem \ref{Problem5} for a fixed $\lambda$. Further, it has a threshold structure, i.e., there exists $\tau := \tau(\lambda,A,K_{W})$ such that
\begin{align}\label{Det_Th_policy}
a = \left\{ \begin{array}{ll}
1, &\Delta \ge \tau, \\ 
0, & \Delta < \tau. \\
\end{array}
\right.
\end{align}
\end{theorem}}
\begin{proof}
The proof of Theorem \ref{Theorem_det} requires some preliminary build-up.  
Similar to the time-average cost function as in \eqref{cost-to-go}, we introduce an $\alpha$--discounted cost function on an infinite horizon, with $0 <  \alpha < 1$, as follows:
\begin{align}\label{Disc_Cost}
V_{\alpha}(\Delta,\gamma^d) = \sum_{k=0}^{\infty}\alpha^k C(\Delta[k],a[k]),~\Delta[0] = \Delta.
\end{align}

\begin{definition}
A scheduling policy $\gamma^d$ of the MDP $\operatorname{M}$ is $\alpha$--optimal if it infimizes the total $\alpha$--discounted cost $V_{\alpha}(\Delta,\gamma^d)$ in \eqref{Disc_Cost}. Further, we define $V_{\alpha}(\Delta) = \inf_{\gamma^d \in \Gamma^d} V_{\alpha}(\Delta,\gamma^d)$ as  the optimal discounted cost.
\end{definition}
With the above definition of $V_{\alpha}(\Delta)$, we can write down the Bellman equation that it satisfies, as in the following proposition.

\begin{proposition}\label{Prop_VI_Convergence}
The optimal $\alpha$--discounted cost $V_{\alpha}(\Delta)$ satisfies the following discounted-cost Bellman equation
\begin{align}\label{discounted_Bellman}
V_{\alpha}(\Delta) = \min_{a \in \{0,1\}} \left[ C(\Delta,a) + \alpha \mathbb{E}[V_{\alpha}(\Delta')]\right],
\end{align}
where $\mathbb{E}[V_{\alpha}(\Delta)] = \sum_{\Delta' \in \{{0},\Delta + 1\}}\mathcal{P}_{\Delta, \Delta'}(a)V_{\alpha}(\Delta')$. A deterministic policy that realizes the minimum in \eqref{discounted_Bellman} is an $\alpha$--optimal policy. Further, we can define the value iteration in terms of $V_{\alpha}^{(j)}(\Delta)$ with $V_{\alpha}^{(0)}(\Delta)=0$ and $V_{\alpha}^{(j+1)}(\Delta) = T(V_{\alpha}^{(j)})(\Delta)$ where $T$ is the Bellman operator given by,
\begin{align}\label{Value_iteration}
T(f)(\Delta) = \min_{a \in \{0,1\}} \left[ C(\Delta,a) + \alpha \mathbb{E}[f(\Delta')]\right].
\end{align}
Then, $\lim\limits_{n \rightarrow \infty}V_{\alpha}^{(n)}(\Delta) \rightarrow V_{\alpha}(\Delta)$. 
\end{proposition}

\begin{proof}
We use the monotonicity property of the Bellman operator to show that the value function is its fixed point and it can be obtained by iterative application of the operator. First we prove monotonicity of the Bellman operator. Consider two functions $f,g:\mathbb{Z}^+ \rightarrow \mathbb{R}$ such that $f(\Delta) \leq g(\Delta)$. Then
\begin{align*}
T(f)(\Delta) 
& \leq  C(\Delta,\hat a) + \alpha \mathbb{E}_{\Delta' \sim \mathcal{P}(\cdot \mid \Delta,\hat a)}[f(\Delta')], \\
& \leq  C(\Delta,\hat a) + \alpha \mathbb{E}_{\Delta' \sim \mathcal{P}(\cdot \mid \Delta,\hat a)}[g(\Delta')] = T(g)(\Delta)
\end{align*}
where $\hat a$ is the minimizer in $T(g)(\Delta)$. Using this monotonicity property, we deduce $V^{(0)}_\alpha \leq V^{(1)}_\alpha \leq \ldots$, and using the monotone convergence theorem, we can conclude $V^{(k)}_\alpha \xrightarrow{k \rightarrow \infty} V^{(\infty)}_\alpha$. Moreover, for any given policy $\gamma$,
$
    V^{(\infty)}_\alpha = \lim_{k \rightarrow \infty} V^{(k)}_\alpha \leq \lim_{k \rightarrow \infty} V^{(k)}_{\alpha,\gamma} = V_{\alpha,\gamma},
$
where $V^{(k)}_{\alpha,\gamma}$ denotes the cost under $\gamma$ for $k$ timesteps. Thus, $V^{(\infty)}_\alpha \leq V_\alpha$. In particular, consider $\gamma$ to be the policy such that $a=1$, for all $k$. Then, clearly, this policy leads to a finite average cost. Thus, the optimal cost $V_{\alpha}$ is finite. Furthermore, using Lebesgue's monotone convergence theorem \cite{folland1999real} and finite action space, we can deduce $V^{(\infty)}_\alpha = T(V^{(\infty)}_\alpha)$. By letting $\gamma^*$ denote the optimal policy, we get
\begin{align*}
    & V^{(\infty)}_\alpha(\Delta[0]) = \mathbb{E} \Big[ \sum_{k=0}^T \alpha^k C(\Delta[k],\gamma^*[k]) + \alpha^T V^{(\infty)}_\alpha(\Delta[T]) \Big] \\
    & \geq \mathbb{E} \Big[ \sum_{k=0}^T \alpha^k C(\Delta[k],\gamma^*[k]) \Big] \xrightarrow{T \rightarrow \infty} V_{\alpha,\gamma^*}(\Delta[0]) \geq V_\alpha(\Delta[0]).
\end{align*}
Since we have shown $V^{(\infty)}_\alpha \leq V_\alpha$ we obtain $V^{(\infty)}_\alpha = V_\alpha$.
\end{proof}
\begin{proposition}\label{Prop_Monotonicity}
The optimal $\alpha$--discounted cost $V_{\alpha}(\Delta)$ is monotonically non-decreasing in $\Delta$.
\end{proposition}
The proof follows from the monotonicity property of the Bellman operator as shown in the proof of Proposition \ref{Prop_VI_Convergence}.
\begin{proposition}\label{Proposition_Det_policy}
There exists an $\alpha$--optimal stationary deterministic policy for a fixed $\lambda$. Further, it has a threshold structure, i.e., there exists $\tau := \tau(\lambda,A,K_{W})$ such that
\begin{align}\label{Det_Threshold_policy}
a_{\alpha}(\Delta) = \left\{ \begin{array}{ll}
1, &\Delta \ge \tau, \\ 
0, & \Delta < \tau,
\end{array}
\right.
\end{align}
where $a_{\alpha}$ denotes the action under an $\alpha$--optimal policy.
\end{proposition}

\begin{proof}[Proof of Proposition \ref{Proposition_Det_policy}]
Consider first the case where the optimal action is $a_{\alpha}(\Delta_{1}) = 1$ at state $\Delta_{1}$. Then, from \eqref{discounted_Bellman}, we get $C(\Delta_{1},0) + \alpha V_{\alpha}(\Delta_{1} + 1) \geq C(\Delta_{1}, 1) + \alpha V_{\alpha} ({0}).$
By substituting the value of $C(\cdot,\cdot)$, we have $\alpha V_{\alpha}(\Delta_{1} + 1) \geq \lambda + \alpha V_{\alpha} ({0}).$
Consider now $\Delta_2 \geq \Delta_1$. Then, by Proposition \ref{Prop_Monotonicity}, we have that 
\begin{align*}
\alpha V_{\alpha}(\Delta_2 + 1) &\geq \alpha V_{\alpha}(\Delta_1 + 1) \geq \lambda + \alpha V_{\alpha} ({0}),
\end{align*}
which implies that the optimal policy for state $\Delta_2$ is to schedule the agent. Similar to the first case, for the one where the optimal action is $a_{\alpha}(\Delta) = 0$, we have that for any $\Delta_2 < \Delta_1$, the optimal action satisfies $a_{\alpha}(\Delta) = 0$. Finally, if we start with $a_{\alpha}(\Delta) = 0$ and continue idling, there must exist a $\Delta$ such that we land in the first case. This is due to monotonicity of $V_{\alpha}(\Delta)$. Thus, there exists a threshold $\tau$ as in \eqref{Det_Threshold_policy}, such that for any state $\Delta \geq \tau$, $a_{\alpha}(\Delta) = 1$ is optimal, while for $\Delta < \tau$, $a_{\alpha}(\Delta) = 0$ is optimal, which completes the proof.
\end{proof}
Finally, to complete the proof of Theorem \ref{Theorem_det}, we make use of the main result of \cite{sennott1989average}. Assumptions 1 and 2 in \cite{sennott1989average} follow due to Propositions \ref{Prop_VI_Convergence} and \ref{Prop_Monotonicity}. By using \eqref{discounted_Bellman} we can show that
$V_\alpha(\Delta) \leq C(\Delta,1) + V_\alpha({0})$, which implies $V_\alpha(\Delta) - V_\alpha({0}) \leq C(\Delta,1) =: M_\Delta < \infty$. Moreover, for a given $\Delta$, $\sum_{\Delta'} P_{\Delta \Delta'}(0) M_{\Delta'} = M_{\Delta+1} < \infty$ and $\sum_{\Delta'} P_{\Delta \Delta'}(1) M_{\Delta'} = M_{0} < \infty$. Hence, Assumption 3* of \cite{sennott1989average} is also satisfied. Consider now a sequence of $\alpha_l$ such that $\lim_{l \rightarrow \infty}\alpha_l = 1$. Then, the optimal policy $a_{\alpha_l}$, which is $\alpha_l$--optimal from Proposition \ref{Proposition_Det_policy}, converges to the $g(\Delta)$--optimal policy \cite{sennott1989average}, which verifies the threshold structure of the latter. Further, the time-average cost $\sigma^* = \lim_{\alpha \rightarrow 1}(1-\alpha)V_{\alpha}(\Delta)$ is finite and independent of the initial state \cite{sennott1989average}. This, then completes the proof of the theorem.
\end{proof}

\subsection{Analytical Characterization of $\tau$}
Having established the threshold structure of the $g(\Delta)$--optimal policy for the single-agent scheduling Problem \ref{Problem5}, we can restrict our attention to the finite-state MDP with the state space $\mathcal{S}' = \{0, 1, \cdots, \tau\}$. Then, the one-stage cost for the time-average cost function satisfies the Bellman's equation given by $V(\Delta) + \sigma^* = \min_{a \in \{0,1\}} \left[ C(\Delta,a) + \mathbb{E}[V(\Delta')]\right],$ which can be equivalently written as:
\begin{align}\label{Test_HJB}
V(\Delta) +\sigma^* \hspace{-0.1cm} = \min
\left\lbrace C(\Delta,0) \hspace{-0.1cm} + \hspace{-0.1cm} V(\Delta +1), C(\Delta,1) +V(0)\right\rbrace.
\end{align}

Next, we calculate an analytical expression for finding $\tau$ as a function of the Lagrange multiplier and the system parameters. We know from Theorem \ref{Theorem_det} that $a = 1$ is optimal at $\Delta = \tau$. Then, from \eqref{Test_HJB}, we have $C(\tau,1) + V(0) -\sigma^* < C(\tau,0) + V(\Delta +1) -\sigma^*$,
which yields
\begin{align}\label{test1}
V({0}) +\lambda < V(\tau + 1).
\end{align}

Further, at $\Delta = \tau - 1$, $a=0$ is optimal. Thus, by the same argument, we have that $V(\tau) \leq \lambda + V({0})$,
which on combining with \eqref{test1} leads to $V(\tau) \leq \lambda + V({0}) < V(\tau + 1).$
Furthermore, by using $a=1$ at $\Delta = \tau$, we get 
\begin{align}\label{test3}
V(\Delta) = g(\Delta,A,K_W) + \lambda + V({0}) -\sigma^* .
\end{align}
Since $V(\Delta)$ is monotonically non-decreasing in $\Delta$, $\exists \eta \in [0,1)$ such that $V(\tau + \eta) = \lambda + V({0})$, and by using \eqref{test3},
\begin{align}
\sigma^* = g(\tau + \eta,A,K_W).
\end{align}
Further, for $\Delta < \tau$, we have $V(\Delta + 1)  - V(\Delta) = \sigma^* -g(\Delta,A,K_{W})$,
which on summing from $\Delta = {0}$ to $\Delta = \tau -1$ gives
\begin{align}\label{test4}
V(\tau)=V({0}) +\sigma^*\tau- \sum_{\Delta = {0}}^{\tau-1} {g(\Delta,A,K_{W})}.
\end{align}
Next, we rewrite the expression of $g(\cdot,\cdot,\cdot)$ as
\begin{align}\label{Frobenius_exp}
        & g(\Delta,A,K_{W})  = \sum_{r=1}^{\Delta}tr((A^{r-1})^\top A^{r-1}K_W)\Delta \nonumber \\
    & =\!\sum_{r=1}^{\Delta}tr\!\left(\!\left(A^{r-1}\!K_W^{0.5}\right)^\top\! \!\!A^{r-1}\!K_W^{0.5}\right)\!\Delta\! = \!\sum_{r=1}^{\Delta}\|A^{r-1}K_W^{0.5}\|_F^2 \Delta.
\end{align}
Substituting the above expression \eqref{Frobenius_exp} for $g(\cdot, \cdot, \cdot)$ in \eqref{test4}, we can calculate the value of $\tau$ by using \eqref{test3} and \eqref{test4}.

Next, we provide a simplified expression to calculate $\tau$ for scalar systems ($m=1$ in \eqref{system}).
\subsubsection{Scalar systems}
Equating the values of $V(\tau)$ from \eqref{test3} and \eqref{test4}, we arrive at the equation $(\tau+1) g(\tau + \eta,A,K_W) - \sum_{l=0}^{\tau}g(l,A,K_W)  = \lambda.$
By substituting the expression for $g(\cdot,\cdot,\cdot)$ in this equation yields
\begin{align}\label{calculation_tau}
\left\{ \begin{array}{ll}
f_1(\tau,A,K_W,\lambda) = 0, & A \neq 1,\\ 
f_2(\tau,K_W,\lambda) = 0, & A = 1, \\
\end{array} 
\right.
\end{align}
where
\begin{align*}
 f_1(\cdot, \cdot, \cdot, \cdot) =& K_W\left[(\tau+1) (\tau + \eta)\frac{A^{2\tau + 2\eta}-1}{A^2-1} \right.\\ & \left. + \frac{ \frac{\tau(\tau+1)}{2} - \frac{A^2}{(A^2-1)^2}+\frac{A^{2\tau+2}(-\tau A^2 + \tau + 1)}{(A^2-1)^2}}{A^2-1}\right] \!-\! \lambda, \\
 \!\!f_2(\cdot, \cdot, \cdot) = &{K_W}\left[ (\tau+1) (\tau + \eta)^2 - \frac{\tau(2\tau+1)( \tau + 1)}{6}\right] \!-\! \lambda.
\end{align*}
We note that \eqref{calculation_tau} is an implicit equation in $\tau$ and $\eta$ for given values of $\lambda$, $A$ and $K_W$. Thus, the value of $\tau$ can be calculated by plotting $\eta$ vs $\tau$, and choosing the integer value of $\tau$ for an admissible $\eta$.

\subsection{Multi-agent Randomized Scheduling Policy} \label{subsec:rand_sched_pol}
In the previous subsection, we showed the existence of a single-agent stationary deterministic policy for a fixed $\lambda$. In this subsection, we return to the multi-agent case and obtain the optimal value of $\lambda$. Consequently, we use the threshold characterization of the deterministic policy to obtain the optimal solution to Problem \ref{Problem3}. The latter policy will be a stationary randomized policy because in general, a stationary deterministic policy for a constrained optimization problem as in Problem \ref{Problem3} may not exist \cite{sennott1993constrained}. Henceforth, we also resume the use of superscript $i$ to denote the agent index.


Consider the threshold for the $i^{th}$ agent given by $\tau^i(\lambda) := \tau^i(\lambda,A_i, K_{W^i})$. Then the agent is connected to its respective controller at every $(\tau^i(\lambda)+1)$--th instant. Or equivalently, its update rate can be given by the quantity $\frac{1}{\tau^i(\lambda)+1 }$. Thus, under the constraint \eqref{AoI-constraint2}, we have that 
\begin{align}\label{soft_constraint}
W(\lambda) :=\sum_{i=1}^N \frac{1}{\tau^i(\lambda)+1} \leq R_{d}.
\end{align}

In order to find the optimal value of the Lagrange multiplier solving Problem \ref{Problem3}, we use the Bisection search procedure \cite{chen2021minimizing}, which we summarize next. 
Since $\lambda \geq 0$, we start by initializing two parameters $\lambda_l^{(0)} = 0$ and $\lambda_u^{(0)} = 1$. We then calculate the threshold parameters $\tau^i(\lambda_u^{(0)})$ for all $i$, by using the piece-wise definition in \eqref{calculation_tau}. Consequently, we iterate by setting $\lambda_l^{(j+1)} = \lambda_u^{(j)}$ and $\lambda_u^{(j+1)} = 2\lambda_u^{(j)}$ until the constraint \eqref{soft_constraint} is satisfied for $\lambda_u^r$, for some integer $r$. Then, we define the interval $[\lambda_l^r,\lambda_u^r]$. This interval contains the optimal value of the multiplier $\lambda^*$, which can be calculated using the \emph{Bisection method}. The iteration stops when $|\lambda_u^{(m)} -\lambda_l^{(m)}|\leq \epsilon$, for the iterating index $m$ and for a suitably chosen $\epsilon > 0.$ 

We next construct the stationary randomized policy solving Problem \ref{Problem3}. Define $\lambda_l^* = \lambda_l^{(m)}$ and $\lambda_u^*= \lambda_u^{(m)}$ as obtained above. 
Further, let the stationary deterministic policies $\gamma^{d,i}_{D1}$ and $\gamma^{d,i}_{D2}$ as obtained from Theorem \ref{Theorem_det} be those  corresponding to $\lambda_l^*$ and $\lambda_u^*$, respectively, where $\lambda_l^* \mapsto \tau_l(\lambda_l^*):= \{\tau_l^1(\lambda_l^*), \cdots, \tau_l^N(\lambda_l^*)\}^\top$ and $\lambda_u^* \mapsto \tau_u(\lambda_u^*):= \{\tau_u^1(\lambda_u^*), \cdots, \tau_u^N(\lambda_u^*)\}^\top$. Also, we define $ R_d^u$ and $R_d^l$ as the total bandwidth used corresponding to the multipliers $\lambda_u^*$ and $\lambda_l^*$, respectively. We note that $\tau_l^i$ differs from $\tau_u^i$ by at most one state. Then, we define the probability $p$ and the deterministic policies for all $i$ as:
\begin{align}
\label{Prob_of_randomization} p& := \frac{R_{d} - R_d^u}{R_d^l - R_d^u}, \\
\label{deterministic_policy1}\gamma^{d,i}_{D1}(\Delta^i) & :=  \left\{ \begin{array}{ll}
1, & \Delta^i \ge \tau^i_l(\lambda_l^*,A_i,K_{W^i}) \\ 
0, & \Delta^i < \tau^i_l(\lambda_l^*,A_i,K_{W^i}) \\
\end{array},
\right. \\
\label{deterministic_policy2}\gamma^{d,i}_{D2}(\Delta^i) & :=  \left\{ \begin{array}{ll}
1, & \Delta^i \ge \tau^i_u(\lambda_u^*,A_i,K_{W^i}) \\ 
0, & \Delta^i < \tau^i_u(\lambda_u^*,A_i,K_{W^i}) \\
\end{array}.
\right.
\end{align}
Let $\gamma^{d}_R := [\gamma^{d,1}_{R}, \cdots, \gamma^{d,N}_{R}]^\top$. The randomized policy $\gamma^{d}_R$ for the relaxed Problem \ref{Problem3} can then be obtained as: 
\begin{align}\label{Randomized_policy}
\gamma^{d,i}_R = p \gamma^{d,i}_{D1} + (1-p) \gamma^{d,i}_{D2}, ~\forall i.
\end{align}

Next, we present the following proposition to prove that the randomized policy obtained is indeed optimal for Problem \ref{Problem3}.

\begin{proposition}[Optimality of Randomized Policy]\label{Optimality_of_Randomized_policy}
Under Assumption \ref{As1}, the policy in \eqref{Randomized_policy}, which randomizes over the deterministic policies \eqref{deterministic_policy1} and \eqref{deterministic_policy2} with probability as in \eqref{Prob_of_randomization} is optimal for the relaxed minimization Problem \ref{Problem3}. 
\end{proposition}

\begin{proof}
Consider the MDP $\operatorname{M}$ in Section \ref{sec:DSP}, and define the costs $C_1(\Delta^i,a^i):= g(\Delta^i,A_i,K_{W^i})$ and $C_2(\Delta^i,a^i):=\lambda a^i$. Further, let the expected $C_1$--cost and $C_2$--cost under a policy $\gamma^{d,i}$ be $\bar{C_1}(\gamma^{d,i})$ and $\bar{C_2}(\gamma^{d,i})$, respectively for agent $i$.

Let $\mathbb{S}^i$ be a non-empty set and $\Gamma(\delta^i,\mathbb{S}^i)$ be the class of policies for agent $i$ with the following properties: 
\begin{itemize}
    \item Starting at $\Delta^i[0] = \delta^i$, the state of the MDP ($\Delta^i[k]$) at time $k$, reaches $\mathbb{S}^i$ with probability 1 under the policy $\gamma^{d,i}$, for some k. More precisely, $\mathcal{P}_{\gamma^{d,i}}(\Delta[k] \in \mathbb{S} \mid \Delta[0] = \delta) = 1$, for some $k \geq 1 $.
    \item The expected time of the first passage from $\delta^i$ to $\mathbb{S}^i$ under $\gamma^{d,i}$ is finite.
    \item The corresponding expected $C_1$--cost $\bar{C}_1^{\delta,G}(\gamma^{d,i})$ and the expected $C_2$--cost $\bar{C}_2^{\delta,G}(\gamma^{d,i})$ of the first passage time from $\delta^i$ to $\mathbb{S}^i$ under $\gamma^{d,i}$ are finite.
\end{itemize}
The above properties hold for all $\{\mathbb{S}^i\}_{1 \leq i \leq N}$ and $\{\Gamma(\delta^i,\mathbb{S}^i)\}_{1 \leq i \leq N}$.

Next, we start by listing and subsequently verifying the assumptions in \cite{sennott1993constrained} for our system.

\begin{enumerate}
    \item \textit{For all $\beta > 0$, the set $\mathbb{S}^i(\beta):= \{\Delta^i \mid \exists a^i ~s.t.~ C_1(\Delta^i,a^i) + C_2(\Delta^i,a^i) \leq \beta\}$ is finite $\forall i$:}
    
    Indeed for the underlying Problem \ref{Problem4}, we have that $C_1(\Delta^i,a^i) + C_2(\Delta^i,a^i) = g(\Delta^i,A_i,K_{W^i}) + \lambda a^i \geq g(\Delta^i,A_i,K_{W^i})$. This implies that for any $\Delta^i \in \mathbb{S}^i(\beta)$, the inequality $g(\Delta^i,A_i,K_{W^i}) \leq \beta$ must hold. Then, by Lemma \ref{L3}, we conclude that $\mathbb{S}^i(\beta)$ is always finite. Note here that $\beta^i$ in general can be different for each agent and consequently, we can let $\beta:= \max\limits_{1 \leq i \leq N}\{\beta^i\}$.
    
    \item \textit{For each agent $i$, there exists a stationary policy $\tilde{\gamma}^i$ such that the induced Markov chain has the property that the state space $\mathcal{S}^i$ of the MDP $\operatorname{M}$ consists of a single non-empty positive recurrent class $R^i_+$ and a set $\mathcal{T}^i$ of transient states such that $\tilde{\gamma}^i \in \Gamma(\delta^i,\mathbb{S}^i)$ for $\delta^i \in \mathcal{T}^i$. Moreover, both $\bar{C}_1(\tilde{\gamma}^i) < \infty$ and $\bar{C}_2(\tilde{\gamma}^i) < \infty$:}
    
    Note first that in our case, $\mathcal{S}^i = \mathbb{Z}^+$. Consider the policy $\tilde{\gamma}^i$ to be the same as the relaxed policy $\gamma^{d,*}_R.$ Then, all the states in the set $\{0, \cdots, \tau^i\}$ communicate with the state $\Delta^i = {0}, \forall i$, and for any starting state $\delta^i \in \{0, \cdots, \tau^i\}$. From this, we conclude that $\Delta^i = {0}$ is positive recurrent. Consequently, the set $R^i_+ = \{0, \cdots, \tau^i\}$ forms the required positive recurrent class. Next, the set $\mathcal{T}^i = \mathbb{Z}^+ \setminus R^i_+$. Then, if the initial condition lies in $\mathcal{T}^i$, the state returns to $R^i_+$ under $\tilde{\gamma}^i$, and hence $\tilde{\gamma}^i \in \Gamma(\delta^i,\mathbb{S}^i)$. Further, the cost $\bar{C}_2(\tilde{\gamma}^i)$ is the expected transmission rate, and the cost $\bar{C}_1(\tilde{\gamma}^i)$ is the expected WAoI for agent $i$, both of which are finite. Thus, the sum of expected costs over all agents is also finite.
    
    \item \textit{For each $i$, given two states $\delta^i_1 \neq \delta^i_2$, there exists a policy $\tilde{\gamma}^i$ such that $\tilde{\gamma}^i \in \Gamma(\delta_1^i, \delta_2^i)$:}
    
    Indeed, since the state $\Delta^i_k$ increases by $1$ each time the agent is not connected over the channel and decreases to unity as soon as it is, we can always find a policy that induces a Markov chain such that any state $\delta_2^i$ can be reached from $\delta_1^i$ with positive probability. Furthermore, the first passage times and the corresponding expected costs are also finite.
    
    \item \textit{For each agent $i$, if a stationary policy ${\gamma}^i$ has at least one positive recurrent state, then it has a single positive recurrent class $R^i_+$. 
    }
    
    Indeed, for any policy $\gamma^i$, by Lemma \ref{L3} the cost function is an increasing function of $\Delta^i$ and decreases only when an agent is connected over the downlink, which corresponds to $\Delta^i = {0},~\forall i$. Thus, any positive recurrent class contains $\Delta^i = {0}$ for all $i$, which is then the only possible positive recurrent class. 
    
    \item \textit{There exists a policy $\gamma^i$ such that $\bar{C}_1(\gamma^i) < \infty$ and $\bar{C}_2(\gamma^i) < \alpha$:}
    
    Indeed, since $\bar{C}_2(\gamma^i)$ is the expected transmission rate for each agent, we can choose a policy with $\tau^i(\lambda)$ large enough such that the total sum of costs across all agents is less than $R_d$. Furthermore, we can also verify that $\bar{C}_1(\gamma^i) < \infty, ~\forall i$, and so is then their sum over all agents.
\end{enumerate}

Since we have verified all the necessary assumptions, we use Proposition 3.2 and Lemmas 3.4, 3.7, 3.9 and 3.10 from \cite{sennott1993constrained} to prove the result. To this end, we define $R_d^{\lambda}$ as the expected transmission rate associated with policy $\tau(\lambda)$ and $\lambda^*:= \inf\{\lambda >0 \mid R_d^{\lambda} \leq R_d\}$. A policy is $\tau^*$--optimal if it is optimal for the MDP of each agent with $\tau(\lambda) = \tau(\lambda^*)$.

Since from the bisection search procedure, we have that $\lambda^*_u \searrow \lambda^*$ and $\lambda^*_l \nearrow \lambda^*$, the policies $\tau^i_u(\lambda^*_u) \searrow \tau^i(\lambda^*)$ and $\tau^i_l(\lambda^*_l) \nearrow \tau^i(\lambda^*)$ and are both $\tau(\lambda^*)$--optimal (Lemmas 3.4, 3.7 of \cite{sennott1993constrained}). Since the Markov chains induced by the policies $\gamma^{d,i}_{D1}$ and $\gamma^{d,i}_{D2}$ are both irreducible and $\Delta^i = {0}$ is positive recurrent for all $i$, we can choose either of these when $\Delta^i ={0}$ independently, without affecting the optimality (Proposition 3.2 and Lemma 3.7 of \cite{sennott1993constrained}). Thus, if we randomize over these policies, according to the statement of the current Proposition \ref{Optimality_of_Randomized_policy} with the probability $p$ chosen to satisfy the rate constraint \eqref{AoI-constraint2} with equality, we conclude that the randomized policy $\gamma^{d}_R$ is optimal for Problem \ref{Problem3}.
The proof is thus complete.
\end{proof}
As a result of above proof, we will henceforth denote the optimal randomized policy in \eqref{Randomized_policy} as $\gamma^{d,*}_R$.

\subsection{Multi-agent Scheduling with Hard Bandwidth Constraint}\label{subsec:hard_cons_pol}
In this subsection, we  return to our original Problem \ref{Problem2} with the hard bandwidth constraint and construct a new policy satisfying \eqref{AoI-constraint1}. Let $\gamma^{d,*,i}_R$ be the optimal policy as in \eqref{Randomized_policy}, and $a^i[k]$ be the scheduling decision for agent $i$ under the relaxed problem. Further, define the set $\Omega[k]:= \{i \in [N] \mid a^i[k] = 1\}$, which denotes the agents to be scheduled at time $k$ with the optimal policy and let $\Omega_k$ be the cardinality of $\Omega[k]$. Then, the scheduling decision ${\zeta}^{d,i}[k]$ under the new policy ${\gamma}^{d,i}$, with the hard bandwidth constraint is given as \cite{tang2020scheduling}:
\begin{itemize}
\item If $\Omega_k \leq R_d$, then ${\zeta}^{d,i}[k] = 1, ~\forall a^i[k] = 1$.
\item If $\Omega_k > R_d$, then ${\zeta}^{d,i}[k] = 1$ for a subset $\Omega^s[k] \subset \Omega[k]$ of the agents, which are selected at random (with uniform probability) by the BS such that $\Omega^s_k = R_d$, where $\Omega^s_k$ denotes the cardinality of $\Omega^s[k]$. The agents in the set $\Omega[k]\setminus \Omega^s[k]$ are not selected because of the bandwidth limit.
\end{itemize}

In the following, we show that the costs under $\gamma^{d,*}_R$ and ${\gamma}^{d}$ approach each other asymptotically (as $N \rightarrow \infty$), where ${\gamma}^{d}:=[\gamma^{d,1}, \cdots, \gamma^{d,N}]^\top$.
To this end, we start by defining a new policy $\hat{\gamma}^d$, which transmits exactly the same agents as the relaxed policy, except that, for each additional agent that was not supposed to be transmitted by the hard-bandwidth policy ${\gamma}^d$, it adds an additional penalty to the cost defined as 
\begin{align}\label{Age_Penalty}
    \omega(\Delta,A,K_W) =&  \sum_{l=0}^{\infty}(1-\frac{R_d}{\Omega_k})^{l+1}g(\tau + l,A,K_{W}) \nonumber \\
    & \times \mathbb{I}\{(1-\frac{R_d}{\Omega_k}) >0\}\mathbb{I}\{\Delta \geq \tau\},
\end{align}
where $\mathbb{I}{\{\cdot\}}$ denotes the indicator function of its argument. Further, we let $\{\tilde{\Delta}^i[1], \tilde{\Delta}^i[2], \cdots, \tilde{\Delta}^i[k], \cdots \},$ and $\{\Delta^i[1],\Delta^i[2], \cdots, \Delta^i[k], \cdots\}$ be the sequences of AoIs of the $i^{th}$ agent under ${\gamma}^{d,i}$ and $\gamma^{d,*,i}_R$ (or equivalently $\hat{\gamma}^{d,i}$), respectively.

Then, we state the following lemma without proof.

\begin{lemma}\label{L5}
The age penalty $\omega(\tilde{\Delta}^i[k],A_i,K_{W^i})$ 
as in \eqref{Age_Penalty} dominates the expected cost under WAoI $g(\tilde{\Delta}^i[k],A_i,K_{W^i})$, $\forall i,k$. \hfill \qedsymbol
\end{lemma}

As a consequence of Lemma \ref{L5}, it immediately follows that
\begin{align}\label{Cost_comparison}
    J^S_{\gamma^{d,*}_R} \leq J^S_{\gamma^{d,*}} \leq J^S_{{\gamma}^d} \leq J^S_{\hat{\gamma}^d}, 
\end{align}
where $\gamma^{d,*}$ is any optimal policy that solves Problem \ref{Problem2}.

Consider next the Markov chain induced by the relaxed policy $\gamma^{d,*,i}_R$ for the $i^{th}$ agent as

\begin{align*}
\Delta^i[k+1] \!=\! \left\{ \begin{array}{ll}
\Delta^i[k] + 1, \qquad \text{w.p.}~ 1, &  \Delta^i[k] < \tau_l^i(\lambda^*_l), \\
\!\!\left\{ \begin{array}{ll}
\Delta^i[k] + 1, & \text{w.p.}~ 1-p, \\
{0}, &\text{w.p.}~ p
\end{array}
\right., & \Delta^i[k] =\tau_l^i(\lambda^*_l), \\
{0}, \qquad \text{w.p.}~ 1, & \Delta^i[k] =\tau_u^i(\lambda^*_u).
\end{array}
\right.
\end{align*}
The state transition diagram for the above dynamics is shown in Figure \ref{Fig:Markov_chain}.
Let $\pi^i[k] := [\pi^i_{0}[k], \cdots, \pi^i_{\tau^i_u(\lambda^*_u)}[k]]^\top$ denote the probability distribution of the states in the above mentioned Markov chain such that $\pi^i_\ell[k] \triangleq \operatorname{Pr}(\Delta^i[k] = \ell), \forall \ell$. Then, we can write the Markov chain dynamics using the generator as
\begin{align}\label{Markov_chain}
    \!\!\!\!\begin{bmatrix}
    \pi^i_0[k\!+\!1] \\
    \pi^i_1[k\!+\!1] \\
    \vdots \\
    \pi^i_{\tau_l^i(\lambda^*_l)}[k\!+\!1] \\
    \pi^i_{\tau_u^i(\lambda^*_u)}[k\!+\!1] 
    \end{bmatrix}\! \!=\! \!\begin{bmatrix}
    0 &0 & \cdots & p & 1 \\
    1 &0 & \cdots & 0 & 0  \\
    0 &1 & \cdots & 0 & 0  \\
     \vdots &\vdots & \ddots & \vdots &\vdots \\
    0 &0 & \cdots &1-p & 0
    \end{bmatrix} \!\! \begin{bmatrix}
    \pi^i_0[k] \\
    \pi^i_1[k] \\
    \vdots \\
    \pi^i_{\tau_l^i(\lambda^*_l)}[k] \\
    \pi^i_{\tau_u^i(\lambda^*_u)}[k]
    \end{bmatrix}.
\end{align}

\begin{figure}[h!]
    \centering
	\begin{tikzpicture}[->, >=stealth', auto, semithick, node distance=1.7cm]
	\tikzstyle{every state}=[fill=white,draw=black,thick,text=black,scale=1,minimum size=1.3cm]
	\node[state]    (A)                     {$0$};
	\node[state]    (B)[right of=A]   {$1$};
	\node[right of=B]   (C){$\cdots$};
	\node[state]    (D)[right of=C]   {$\tau^i_l(\lambda^*_l)$};
	\node[state]    (E)[right of=D]   {$\tau^i_u(\lambda^*_u)$};
	\path
	(A) edge[bend left,above]	node{$1$}	(B)
	(B) edge[bend left,above]	node{$1$}	(C)
 	(C) edge[bend left,above]	node{$1$}	(D)
	(D) edge[bend left,above]	node{$1-p$}	(E)
	edge[bend left, below]	    node{$p$}	(A)
	(E) edge[bend left, below]	node{$1$}	(A.south);
	\end{tikzpicture}
	\vspace{-0.35cm}
    \caption{Markov chain induced by the policy $\gamma^{d,*,i}_R$ for agent $i$}
        \label{Fig:Markov_chain}
        \vspace{-0.5cm}
\end{figure}
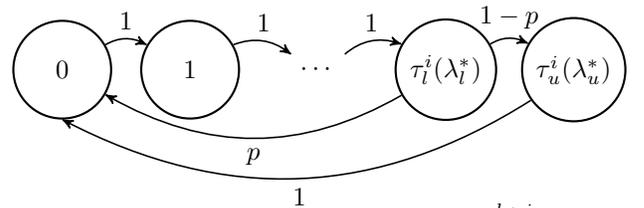

Define $\pi^i := [\pi^i_0, \cdots, \pi^i_{\tau^i_u(\lambda^*_u)}]^\top$. Then, since each state in the set $\mathcal{S}^{u,i} := \{{0}, 1, 2, \cdots, \tau_u^i(\lambda^*_u)\}$ is reachable from every other state, the Markov chain in \eqref{Markov_chain} is irreducible. This implies that it admits a unique stationary distribution $\pi^i$ \cite{norris1998markov}. We now state the following assumption and then present the following main result.
\begin{assumption}\label{AsAbound}
The inequality $0 < \|A(\theta)\|_F < \sqrt{1/(1-\alpha)}$ holds $\forall \theta \in \Theta$, where $\alpha = \frac{R_d}{N}$.
\end{assumption}
We remark that the Assumption \ref{AsAbound} entails finite scheduling cost under the hard-bandwidth policy, and further insights on it are presented later (in Remark \ref{remark_AboundedAssump}).




\begin{theorem}\label{Th:Asymptotic_optimality}
Assume that the proportion of agents $\frac{R_d}{N} = \alpha$ that can be transmitted over the channel is fixed. Then, the deviation of the relaxed scheduling policy $\gamma^{d,*}_R$ from the policy ${\gamma}^{d}$ is of the order of $\mathcal{O}\left(\frac{1}{\sqrt{N}} \right)$.
Consequently, as $N \rightarrow \infty$, ${\gamma}^{d}$ is asymptotically optimal for Problem \ref{Problem2}.
\end{theorem}

\begin{proof}
    Consider the following.
\begin{align}\label{derivation}
    & J^S_{\hat{\gamma}^d} - J^S_{\gamma^{d,*}_R} = \limsup_{T \rightarrow \infty} \frac{1}{NT}\mathbb{E}\bigg[\sum_{k=0}^{T-1}\sum_{i=1}^N g(\Delta^i[k],A_i,K_{W^i}) \nonumber \\
    & \qquad\qquad\!\! + \omega(\tilde{\Delta}^i[k],A_i,K_{W^i})\bigg]\Bigg|_{\hat{\gamma}^d}  \nonumber \\
    & \qquad \qquad \!\!-\limsup_{T \rightarrow \infty} \frac{1}{NT}\mathbb{E}\left[\sum_{k=0}^{T-1}\sum_{i=1}^N g(\Delta^i[k],A_i,K_{W^i})\right]\Bigg|_{\gamma^{d,*}_R} \nonumber \\
    & = \limsup_{T \rightarrow \infty} \frac{1}{NT}\mathbb{E}\left[\sum_{k=0}^{T-1}\sum_{i=1}^N
    \sum_{l=0}^{\infty}\!\left(1-\frac{R_d}{\Omega_k}\right)^{l+1}\!\!\!\!\!\!g(\tau^i + l,A_i,K_{W^i}) \right.  \nonumber \\ & \left. \qquad \qquad \times \mathbb{I}\left\lbrace\left(1-\frac{R_d}{\Omega_k}\right) >0\right\rbrace\mathbb{I}\left\lbrace\tilde{\Delta}^i[k] \geq \tau^i\right\rbrace\right]  \nonumber \\
    & \leq \limsup_{T \rightarrow \infty} \frac{1}{NT}\mathbb{E}\left[\sum_{k=0}^{T-1}\sum_{i=1}^N
    \sum_{l=0}^{\infty}\!\left(1-\frac{R_d}{\Omega_k}\right)^{l+1}\!\!\!\!\!\!g(\tau^i + l,A_i,K_{W^i}) \right.  \nonumber \\ & \left. \qquad \qquad \times \mathbb{I}\left\lbrace\left(1-\frac{R_d}{\Omega_k}\right) >0\right\rbrace\right]  \nonumber \\
    & \leq \limsup_{T \rightarrow \infty} \frac{\mathcal{M}}{NT}\mathbb{E}\left[\sum_{k=0}^{T-1}\sum_{i=1}^N\left\lbrace \frac{\Omega_k \!-\!R_d}{\Omega_k} \right\rbrace^+ \right] \nonumber \\
    & \leq \frac{\mathcal{M}}{N\alpha} \limsup_{T \rightarrow \infty} \frac{1}{T}\sum_{k=0}^{T-1}\mathbb{E}\left[ |\Omega_k-R_d|  \right] \nonumber \\
    &  \leq \frac{\mathcal{M}}{N\alpha} \limsup_{T \rightarrow \infty} \frac{1}{T}\sum_{k=0}^{T-1} \underbrace{\mathbb{E}\left[|\Omega_k - \mathbb{E}\left[ \Omega \right]| \right]}_{:=\operatorname{MAD}(\Omega_k)},
\end{align}
where $J^S_{\hat{\gamma}^d}$ and $J^S_{\gamma^{d,*}_R}$ are the costs under policies $\hat{\gamma}^d$ and $\gamma^{d,*}_R$, respectively. The expression $g(\cdot, \cdot, \cdot)|_{*}$ denotes the WAoI under the policy $*$. The second equality follows since sample paths of the AoI under $\hat{\gamma}^d$ coincide with those under the policy $\gamma^{d,*}_R$ by definition. The second inequality follows as a result of Assumption \ref{AsAbound} and $\mathcal{M}$ is a constant. The third inequality follows since $\{\cdot\}^+ := \max\{\cdot,0\}$ and we use the fact that $\{\cdot\}^+ \leq |\cdot|$.
Further, in the last inequality, we have that $\mathbb{E}\left[ \Omega \right] := \sum_{i \in [N]}\sum_{j \in \mathcal{S}^{u,i}}\pi^i_j\gamma^{d,*,i}_{R,j} = \lim_{T \rightarrow \infty} \frac{1}{T}\sum_{k=0}^{T-1}$ $\sum_{i=1}^N a^i[k]$, which follows as a result of the Ergodic theorem \cite{norris1998markov} due to irreducibility of the Markov chain \eqref{Markov_chain} and $\gamma^{d,*,i}_{R,j}$ being the indicator of the transmission of agent $i$ in state $j$ under the relaxed policy. Further, $\operatorname{MAD}(\cdot)$ stands for the mean absolute deviation of its argument. Note also that $\Omega_k = \sum_{i=1}^Na^i[k]$, where the scheduling decisions $a^i[k] \in \{0,1\}$ are independent binary random variables due to the fact that the relaxed policy was obtained by decoupling the problem into a single agent problem.
Let $\sigma^i :=\sum_{j \in \mathcal{S}^{u,i}}\pi^i_j\gamma^{d,*,i}_{R,j}$ denote the probability of transmission of agent $i$ under $\gamma^{d,*}_R$ in state $j$ under the stationary distribution $\pi^i$. Then, the random variable $Y[k] = \frac{\Omega_k-\sum_{i=1}^N\sigma^i}{\sqrt{\sum_{i=1}^N\sigma^i(1-\sigma^i)}}$ has zero mean and unit variance.

Continuing from \eqref{derivation}, we get
\begin{align}
    J^S_{\hat{\gamma}^d} - J^S_{\gamma^{d,*}_R}  \leq \frac{\mathcal{M}}{\sqrt{N}\alpha} \limsup_{T \rightarrow \infty} \frac{1}{T}\sum_{k=0}^{T-1} \operatorname{MAD}\left( \frac{\Omega_k}{\sqrt{N}}\right) \nonumber
    \end{align}
    \begin{align}\label{tempeqn4}
     \leq \frac{\mathcal{M}}{\sqrt{N}\alpha} \limsup_{T \rightarrow \infty} \frac{1}{T}\sum_{k=0}^{T-1}\sqrt{\frac{1}{N}\times \frac{N}{4}} =\frac{\mathcal{M}}{2\alpha \sqrt{N}},
\end{align}
where the second inequality is true by Jensen's inequality and we use the fact that $a(1-a) \leq 1/4$ for $0 \leq a \leq 1$. Finally, using Lemma \ref{L5} and \eqref{Cost_comparison}, we have that 
\begin{align}\label{tempeqn5}
    J^S_{{\gamma}^d} - J^S_{\gamma^{d,*}_R} \leq J^S_{\hat{\gamma}^d} - J^S_{\gamma^{d,*}_R} \leq \frac{\mathcal{M}}{2\alpha \sqrt{N}},
\end{align}
Thus, as $N \rightarrow \infty$, from  \eqref{tempeqn5}, we have that the hard-bandwidth policy ${\gamma}^d$ is asymptotically optimal for Problem \ref{Problem2} with order $\mathcal{O}(1/\sqrt{N})$. The proof is thus complete.

\end{proof}

Since the solution to the scheduling problem is complete, in the next section, we proceed to establish the $\epsilon$--Nash property of the mean-field solution.

\section{$N$ Agent Consensus Problem}\label{sec:cons_prob}
In this section, we solve the second part of the $N+1$--player game problem, namely, the consensus problem using the central scheduler's policy. To this end, we first consider the limiting game called the mean-field game (MFG) (as $N \rightarrow \infty$). Under this setting, the empirical coupling term in \eqref{LQT} is approximated by a known deterministic sequence (or the MF trajectory) and the closeness of the approximation is justified later in the analysis. This principle is well known in the literature as the Nash certainty equivalence principle \cite{huang2006large} and reduces the game problem to a stochastic optimal control problem of a generic agent with a coupled consistency condition. The equilibrium solution obtained (called the mean-field equilibrium (MFE)) will be shown to constitute an $\epsilon$--Nash approximation to the finite agent game problem.

\subsection{Stochastic Optimal Tracking Control}
Consider a generic agent of type $\theta$ from an infinite population with dynamics
\begin{align}\label{generic_dyn}
    X[k+1] = A(\theta)X[k] + B(\theta)U[k] + W[k],
\end{align}
where timestep $k \in \mathbb{Z}^+$, $X[k] \in \mathbb{R}^n$ and $U[k] \in\mathbb{R}^m$ are the state vector and the control input, respectively. Further, the initial state $X[0]$ is assumed to have a symmetric density function with $\mathbb{E}[X[0]] = \nu_{\theta,0}$ and $cov(X[0]) = \Sigma_x$ is bounded. Next, $W_k \in \mathbb{R}^n$ is a zero-mean i.i.d. Gaussian noise with finite covariance $K_W$. All covariance matrices are assumed to be positive-definite. The objective function of the generic agent is given by
\begin{align}\label{LQT_MF}
    & J(\mu,\bar{X}):= \nonumber \\
    &\limsup_{T \rightarrow \infty} \frac{1}{T}\mathbb{E}\left[ \sum_{k=0}^{T-1}\|X[k] - \bar{X}[k]\|^2_{Q(\theta)} + \|U[k]\|^2_{R(\theta)} \right],
\end{align}
where $\mu := (\mu[1], \mu[2], \cdots, ) \in \mathcal{M}^{d,con}$ and is an admissible control policy of the generic agent. Further, the admissible set $\mathcal{M}^{d,con}:=\{\mu \mid \mu ~\text{is adapted to } \sigma(I^{d,con}[s], s =0,1, \cdots, k)\}$ is the space of \textit{decentralized control} policies for the generic agent and $I^{d,con}[0]:=\{Z[0]\}$, $I^{d,con}[k]:=\{U_{0:k-1},Z_{0:k}\}, k \geq 1$, is the local information history of the generic agent. Recall that this is in contrast to $I^{c,con}[k]$ (the centralized history at the controller), which includes information of other agents as well. This implies that $\mathcal{M}^{d,con} \subseteq \mathcal{M}^{c,con}$. The information structure for the generic agent's decoder is defined similar to that in Subsection \ref{subsec:cons_prob}, except with the superscript $i$ removed. Further, $\bar{X} = (\bar{X}[1], \bar{X}[2], \cdots, )$
is the MF trajectory and denotes the infinite player approximation to the consensus term in \eqref{LQT} and serves to decouple the otherwise coupled game problem into a linear-quadratic tracking (LQT) problem via introducing indistinguishability between agents. Finally, the expectation above is taken with respect to the noise statistics and the initial state distribution.

To solve the LQT problem with dynamics in \eqref{generic_dyn} and the cost in \eqref{LQT_MF}, we first state the following assumption.

\begin{assumption}\label{As3}
\begin{enumerate}[(i)]
    \item The pair $A(\theta),B(\theta)$ is controllable and $(A(\theta),\sqrt{Q(\theta)})$ is observable.
    \item The MF trajectory belongs to the space of bounded functions, i.e., $\bar{X} \in \mathcal{X}:= \{\bar{X}[k] \in \mathbb{R}^n \mid \|\bar{X}\|_{\infty}:= \sup_{k \geq 0} \|\bar{X}[k]\| < \infty$\}.
\end{enumerate}
\end{assumption}
We note that Assumption \ref{As3} is standard in the MF-LQG literature \cite{moon2014discrete,aggarwal2022linear} to entail a well-defined solution to the stochastic optimal control problem of the generic agent.

Next, we define a MFE by introducing the following operators \cite{uz2020reinforcement}:
\begin{enumerate}
    \item $\Psi: \mathcal{X} \rightarrow \mathcal{M}^{d,con}$, defined as $\Psi(\bar{X}) = \argmin_{\mu \in \mathcal{M}^{d,con}}$ $J(\mu, \bar{X})$, gives the optimal control policy for a given MF trajectory, and
    \item $\Lambda: \mathcal{M}^{d,con} \rightarrow \mathcal{X}$, also called the consistency operator, regenerates a MF trajectory for a control policy as obtained in 1) given above.
\end{enumerate}
\begin{definition}[Mean-Field Equilibrium (MFE)]
The pair $(\mu^*, \bar{X}^*) \in \mathcal{M}^{d,con} \times \mathcal{X}$ is a MFE if, $\mu^* = \Psi(\bar{X}^*)$ and $\bar{X}^* = \Lambda(\mu^*)$. In other words, $\bar{X}^* = \Lambda \circ \Psi(\bar{X}^*)$.
\end{definition}

Now, the central scheduling policy is fixed from the previous section, similar to \cite{aggarwal2022linear}, we have the following result for the optimal tracking control of a generic agent.
\begin{proposition}\label{optimalTrackControl}
Consider the generic agent \eqref{generic_dyn} with the controller state as in \eqref{Recursive} and the cost function in \eqref{LQT_MF}. Then, under Assumptions \ref{As1}-\ref{As3}, the following hold true:
\begin{enumerate}
    \item The optimal decentralized control action for the Problem \ref{Problem1} is given as
    \begin{align}\label{optimal_control}
        U^*[k] = -\Pi(\theta)Z[k] - L(\theta)r[k+1],
    \end{align}
    where $L(\theta) =   (R(\theta) +B(\theta)^\top K(\theta)B(\theta))^{-1}B(\theta)^\top$, 
$\Pi(\theta) = L(\theta)K(\theta)A(\theta)$,
and $K(\theta)$ is the unique positive definite solution to the algebraic Riccati equation,
\begin{align*}
K(\theta) = A(\theta)^\top [K(\theta)A(\theta) - K(\theta)^\top B(\theta)\Pi (\theta)] + Q(\theta).
\end{align*}
Further, the trajectory $r[k]$ satisfies the backward dynamics $r[k] = H(\theta)^\top r[k+1] - Q(\theta)\bar{X}[k]$,
with the initial condition $r[0] = -\sum_{j=0}^{\infty}{(H(\theta)^j})^\top$ $Q(\theta)\bar{X}[j]$ and $H(\theta) = A(\theta) - B(\theta)\Pi(\theta)$ is Hurwitz.
\item The difference equation for $r[k]$ above has a unique solution in $\mathcal{X}$, given as
\begin{align}\label{bckw_dyn}
    r[k] = -\sum_{j=k}^{\infty}{(H(\theta)^{j-k})}^\top Q(\theta)\bar{X}[j].
\end{align}
\item The optimal cost is bounded above as:
\end{enumerate}
\begin{align}\label{LQGCost}
    & J(\mu^*\!\!,\bar{X}) \!\leq \!tr(K(\theta)K_W(\theta))\! +\! \limsup_{T \rightarrow \infty} \frac{1}{T}\!\sum_{k=0}^{T-1}\!\!\bar{X}[k]^\top \!Q(\theta) \bar{X}[k] \nonumber \\ &  -  r[k+1]^\top B(\theta)L(\theta)r[k+1] \nonumber \\ &  + \|A(\theta)^\top\!\! K(\theta)^\top\!\! B(\theta)\Pi(\theta)\|\sum_{r=1}^{\tau_u}tr(A(\theta)^{{r-1}^\top}\!\!A(\theta)^{r-1}K_W(\theta)) 
    \end{align}
    

\end{proposition}
\begin{proof}
Part 1) follows from \cite{aggarwal2022linear,moon2014discrete}. Further, we note that $r[0] \in \mathcal{X}$ since $H(\theta)$ is Hurwitz, and thus $r[k] \in \mathcal{X}$ as in part 2) and is unique.
To prove part 3), we substitute \eqref{optimal_control} in \eqref{LQT_MF}, to arrive at
\begin{align}\label{temp_costLQ0}
    & J(\mu^*\!\!,\bar{X}) \!\leq \!tr(K(\theta)K_W(\theta))\! +\! \limsup_{T \rightarrow \infty} \frac{1}{T}\!\sum_{k=0}^{T-1}\!\!\bar{X}[k]^\top \!Q(\theta) \bar{X}[k] \nonumber \\ & - r[k+1]^\top B(\theta)L(\theta)r[k+1]  \nonumber \\ & + \limsup_{T \rightarrow \infty} \frac{1}{T}\sum_{k=0}^{T-1}\|A(\theta)^\top K(\theta)^\top B(\theta)\Pi(\theta)\| ~\mathbb{E}\left[\|e_k\|^2 \right].
\end{align}
Then, consider the following:
\begin{align}\label{temp_costLQ}
    & \sum_{k=0}^{T-1}\mathbb{E}\left[\|e_k\|^2 \right] \leq   T\sum_{r=1}^{\tau_u}tr(A(\theta)^{{r-1}^\top}A(\theta)^{r-1}K_W(\theta)) \nonumber \\
    & +  \sum_{k=0}^{T-1} \sum_{r=\tau_u+1}^k tr(A(\theta)^{{r-1}^\top}A(\theta)^{r-1}K_W(\theta))(1-\alpha)^{k-\tau_u-1} \nonumber \\
    & \leq T\sum_{r=1}^{\tau_u}tr(A(\theta)^{{r-1}^\top}A(\theta)^{r-1}K_W(\theta)) \nonumber \\ & + \sum_{k=\tau_u+1}^{T-1} \|K_W(\theta)\|_F\frac{(1-\|A(\theta)\|_F^{2(k-\tau_u-1)})(1-\alpha)^{k-\tau_u-1}}{1-\|A(\theta)\|_F^2} \nonumber \\
    & = T\sum_{r=1}^{\tau_u}tr(A(\theta)^{{r-1}^\top}A(\theta)^{r-1}K_W(\theta)) + \frac{\|K_W(\theta)\|_F}{\|A(\theta)\|_F^2\!-\!1}\nonumber \\
    & \times \left[\frac{1\!-\!(\|A(\theta)\|^2_F(1\!-\!\alpha))^{T\!-\!\tau_u\!-\!1}}{1-\|A(\theta)\|^2_F(1-\alpha)} \!-\! \frac{1\!-\!(1\!-\!\alpha)^{T\!-\!\tau_u\!-\!1}}{\alpha}\right],
\end{align}
where $\tau_u:= \tau_u(\lambda_u^*)$, the first inequality follows using the hard-bandwidth policy of Section \ref{subsec:hard_cons_pol}, and the second  inequality follows using Assumption \ref{AsAbound}, and the fact that $\|AB\|_F \leq \|A\|_F\|B\|_F$. Then, combining \eqref{temp_costLQ0} and \eqref{temp_costLQ}, we arrive at \eqref{LQGCost}.
\end{proof}
\begin{remark}\label{remark_AboundedAssump}
We remark here that while Assumption \ref{AsAbound} entails a finite optimal cost as in \eqref{LQGCost}, it warrants that one requires a higher downlink bandwidth $(R_d)$ to accommodate a higher degree of instability among the agents. This can be seen from the fact that $\rho(A(\theta)) \leq \|A(\theta)\|_F$, where $\rho(\cdot)$ denotes the spectral radius of the argument matrix. Thus, the bound of the Frobenius norm as in Assumption \ref{AsAbound} in turn bounds the maximum eigenvalue that can be stabilized using the optimal control in \eqref{optimal_control}. This can be further observed from the special case of a scalar system, where the Frobenius norm equals the magnitude of the only eigenvalue, and is bounded as a function of the available bandwidth.
\end{remark}

\subsection{Mean-Field Analysis}

In this subsection, we prove the existence and uniqueness of the MFE by explicitly constructing the MF operator as follows.

Consider the closed-loop system in \eqref{Recursive} with the control policy in \eqref{optimal_control} as
\begin{align*}
 Z[k+1] =\hspace{7.5cm}\\  \left\{\begin{array}{ll}
H(\theta)Z[k] - B(\theta)L(\theta)r[k+1] + W[k], &\zeta^{d}[k+1]=1, \\ 
H(\theta)Z[k] - B(\theta)L(\theta)r[k+1], & \zeta^{d}[k+1]=0, \\
\end{array}
\right.
\end{align*}
where $\zeta^d[k]$ is chosen from the hard-bandwidth policy $\gamma^d$ of Section \ref{subsec:hard_cons_pol}.  The above can be rewritten as 
\begin{align*}
    X[k\!+\!1] \!\!=\! H(\theta)X[k] \!\!-\!\! B(\theta)L(\theta)r[k\!+\!1] \!\! +\! B(\theta)\Pi(\theta)e[k] \!\!+\! W[k],
\end{align*}
where $e[k]$ is defined in \eqref{error_definition}. By taking expectation on both sides and substituting $r[k]$ from \eqref{bckw_dyn}, we get
\begin{align}\label{X2}
     \hat{X}_{\theta}[k] =& \sum_{j=0}^{k-1}\!H\!(\theta)^{k-j-1}B(\theta)L(\theta)\!\!\sum_{s=j+1}^{\infty}\!(H\!(\theta)^{s-j-1})^\top \!Q(\theta)\bar{X}[s]\nonumber\\
     &+H\!(\theta)^k \nu_{\theta,0},
\end{align}
where $\hat{X}_{\theta}[k] = \mathbb{E}[X[k]]$ is the aggregate dynamics across agents of type $\theta$ and we use the fact that $\mathbb{E}[e[k]]=0$.

Next, using the empirical distribution from Section \ref{sec:Prob_Form}, we define the MF operator as
\begin{align}\label{MF_Op}
    \mathcal{T}(\bar{X})[k]:= \sum_{\theta \in \Theta}\hat{X}_{\theta}[k]\mathbb{P}(\theta).
\end{align}
Then, we state the following assumption before we prove the main result.

\begin{assumption}\label{As4}
We assume $\check{H}(\theta) := \|H(\theta)\| +\upsilon <1, \forall \theta \in \Theta$, where $\upsilon = \sum_{\theta \in \Theta}{\frac{\|Q(\theta)\|\|B(\theta)L(\theta)\|}{(1-\|H(\theta)\|)^2}\mathbb{P}(\theta)} $.
\end{assumption}
It is common in the literature (\cite{huang2007large,uz2020reinforcement}) to invoke the above assumption; although it is stronger than the corresponding assumption in \cite{uz2020approximate} and \cite{moon2014discrete}, it leads to the linearity property of the MF trajectory, which can then be computed offline.

\begin{theorem}\label{Th:existUniq}
Under Assumptions \ref{AsAbound}-\ref{As4}, the following statements hold true:
\begin{enumerate}[(i)]
\item The operator $\mathcal{T}(\bar{X}) \in \mathcal{X}, ~~\forall \bar{X} \in \mathcal{X}$. Furthermore, there exists a unique $\bar{X}^* \in  \mathcal{X} $ such that $\mathcal{T}(\bar{X}^*) = \bar{X}^*$.
\item $\bar{X}^*[k]$ follows linear dynamics, i.e., $\exists ~ \mathscr{E}^* \in \mathcal{E}:= \{\mathscr{E}\in \mathbb{R}^{n\times n}:~ \|\mathscr{E}\|\leq 1, \bar{X}^*[k+1] = \mathscr{E}\bar{X}^*[k]\}$, where $\bar{X}^*[k]$ is the equilibrium MF trajectory of the agents, and $\bar{X}^*[0] = \sum_{\theta \in \Theta}{\nu_{\theta,0}\mathbb{P}(\theta)}$.
\end{enumerate}
\end{theorem}

\begin{proof}
The proof is similar to that in \cite{aggarwal2022linear}, but we reproduce it here for completeness.
\begin{enumerate}[(i)]
\item Consider the system in \eqref{X2} driven by the bounded input $\bar{X}[k]$. Since $\|H(\theta)\| <1$ by Assumption \ref{As4}, and $r[k] \in \mathcal{X},~\forall k$, we get that $\sup_{k\geq 0}{\|\mathcal{T}(\bar{X})[k]\|} < \infty$, using \eqref{X2} and \eqref{MF_Op}. Thus, $\mathcal{T}(\bar{X}) \in \mathcal{X}$. Next, consider the following:
\begin{align*}
&\|\mathcal{T}(\bar{X}_1)[k]- \mathcal{T}(\bar{X}_2)[k]\| \! = \!\|\!\sum_{\theta \in \Theta}({\hat{X}_{\theta,1}[k] - \hat{X}_{\theta,2}[k])\mathbb{P}(\theta)}\| \\
& \leq \sum_{\theta \in \Theta}\!{(\|Q(\theta)\|\|B(\theta)L(\theta)\|\left(\sum_{s=0}^{\infty}\|H(\theta)\|^s\right)^2\mathbb{P}(\theta)} \\ & \qquad \times \|\bar{X}_1 - \bar{X}_2\|_{\infty}  = \upsilon \|\bar{X}_1 - \bar{X}_2\|_{\infty},
\end{align*} 
where the last equality follows as a result of Assumption \ref{As4}. Taking supremum over $k$ in the LHS of the above equation, we get that $\mathcal{T}(\cdot)$ is a contraction. Further, since the metric space $(\mathcal{X},||\cdot||_{\infty})$ is complete \cite{folland1999real}, using Banach's fixed point theorem, we deduce that $\mathcal{T}(\bar{X})$ has a unique fixed point in $\mathcal{X}$.

\item We define a new operator $\check{\mathcal{T}}: \mathbb{R}^{n\times n} \rightarrow \mathbb{R}^{n\times n}$ as:
\begin{align*}
\check{\mathcal{T}}_{\theta}(\mathscr{E}) & := H(\theta) + B(\theta)\mathscr{E}(\theta)\sum_{\alpha=0}^{\infty}{(H(\theta)^{\alpha})^\top Q(\theta)\mathscr{E}^{\alpha+1}},\\
\check{\mathcal{T}}(\mathscr{E}) & := \sum_{\theta \in \Theta} \check{\mathcal{T}}_{\theta}(\mathscr{E}) \mathbb{P}(\theta),
\end{align*} 
The proof of the existence of fixed point $\mathscr{E}^*$ of $\check{\mathcal{T}}$,  follows in a similar manner as in \cite{uz2020reinforcement}, which gives
\begin{align*}
	& \small{\|\check{\mathcal{T}}(\mathscr{E}_2)\!\! -\!\! \check{\mathcal{T}}(\mathscr{E}_1)\|} \small{<\! \!\sum_{\theta \in \Theta}\!\!\frac{\|B(\theta)
	L(\theta)\|\|Q(\theta)\|}{(1-\|H(\theta)\|)^2}\|\mathscr{E}_2-\mathscr{E}_1\| \mathbb{P}(\theta)}.
\end{align*}
Consequently, under Assumption \ref{As4}, we have that $\check{\mathcal{T}}$ is a contraction. Using completeness of $\mathcal{E}$ and Banach's fixed point theorem, we indeed have the existence of such an $\mathscr{E}^*$. Finally, from part (i) above, we can recursively construct the unique MF trajectory $\bar{X}^*$ as $\bar{X}^*[k] = (\mathscr{E}^*)^k\bar{X}^*[0]$ with $\hat{X}[0] = \nu_{\theta,0}$.

\end{enumerate}
\end{proof}
\begin{remark}
While Theorem \ref{Th:existUniq}(i) gives us a unique MFE, the linearity of the MF trajectory in (ii) yields a feedback control law, linear in the agent's state and the equilibrium trajectory. This makes the computation of this trajectory tractable, which would otherwise have involved a non-causal infinite sum. Further, as a result of the linear MF trajectory, we will (later) be able to compute an explicit convergence rate between the equilibrium trajectory and the coupling term in \eqref{LQT} in Proposition \ref{Approx_behav}.
\end{remark}

\subsection{Closed-loop System Analysis}
In this subsection, we show that the closed-loop system is stable under the MFE solution.
\begin{lemma}\label{Th:CLanalysis}
Suppose that Assumptions \ref{As1}-\ref{As4} hold. Then, the closed-loop system \eqref{system} is uniformly mean-square stable under the decentralized equilibrium control policy \eqref{optimal_control}, i.e., $\sup_{N \geq 1}\max_{1 \leq i \leq N} \limsup_{T \rightarrow \infty}\frac{1}{T}\mathbb{E}\left[ \sum\limits_{k=0}^{T-1}\left\lVert X^{*,i}[k]\right\rVert^2 \right]\! < \infty$.
\end{lemma}

\begin{proof}
We will follow the proof technique of \cite{moon2014discrete}. Before proceeding, we drop the $*$ from the superscripts for ease of notation. Then, by using \eqref{generic_dyn} and \eqref{optimal_control}, we can write the closed-loop system as
\begin{align}\label{clsys}
    X^i[k+1] =& H(\theta_i) X^i[k] + B(\theta_i)\Pi(\theta_i) e^i[k] \nonumber \\ &- B(\theta_i)L(\theta_i)r^i[k+1] + W^i[k].
\end{align}
Equivalently, \eqref{clsys} can be written as 
\begin{align}\label{CL_system}
    X^i[k] =& H(\theta_i)^k X^i[0] + \sum_{p=0}^{k-1}H(\theta_i)^{k-p-1}B(\theta_i)\Pi(\theta_i) e^i[p] \nonumber \\ &- \sum_{p=0}^{k-1}H(\theta_i)^{k-p-1}B(\theta_i)L(\theta_i) r^i[p+1] \nonumber \\ & + \sum_{p=0}^{k-1}H(\theta_i)^{k-p-1}W^i[p].
\end{align}
Then, from \eqref{CL_system} we have that
\begin{align}
     &\mathbb{E}\left[\|X^i[k]\|^2 \right] \leq 4\underbrace{\mathbb{E} \left[ \left\lVert H(\theta_i)^k X^i[0] \right\rVert^2 \right]}_{(I)}  \nonumber \\
    & +4\underbrace{\mathbb{E} \left[ \left\lVert \sum_{p=0}^{k-1}H(\theta_i)^{k-p-1}B(\theta_i)\Pi(\theta_i) e^i[p] \right\rVert^2 \right]}_{(II)}  \nonumber \\
    & +4\underbrace{\mathbb{E} \left[ \left\lVert \sum_{p=0}^{k-1}H(\theta_i)^{k-p-1}B(\theta_i)L(\theta_i) r^i[p+1] \right\rVert^2 \right]}_{(III)}  \nonumber 
        \end{align}
    \begin{align}\label{Expected_CL_system}
     &+4\underbrace{\mathbb{E} \left[ \left\lVert \sum_{p=0}^{k-1}H(\theta_i)^{k-p-1}W^i[p] \right\rVert^2 \right]}_{(IV)},
\end{align}
where we used the fact that $\|\sum_{i=1}^n x_i\|^2 \leq n \sum_{i=1}^n\|x_i\|^2$. Next, we consider the first term in \eqref{Expected_CL_system}. Since the matrices $H(\theta_i)$ are mean-square stable from Theorem \ref{Th:existUniq}, from \cite[Theorem 3.9]{costa2006discrete}, there exist constants $\varrho(\theta_i) \geq 1$ and $\varepsilon(\theta_i) \in (0,1)$ such that
\begin{align}\label{Term1}
(I) \leq \varrho(\theta_i)\varepsilon(\theta_i)^k tr(\Sigma_x + \nu_{\theta_i,0}\nu_{\theta_i,0}^\top).
\end{align}
Analogously for the fourth term in \eqref{Expected_CL_system}, we have
\begin{align}\label{Term4}
    &(IV)\leq \!\mathbb{E} \! \left[\sum_{p=0}^{k-1} \left\lVert H(\theta_i)^{k-p-1}W^i[p] \right\rVert^2 \right] \!\leq  \frac{\varrho(\theta_i)tr(K_0)}{1-\varepsilon(\theta_i)},
\end{align}
where  the first inequality follows by the independence of $W^i[k]$ and the second one follows from the same reasoning as for the first term above.

Next, we consider the third term in \eqref{Expected_CL_system}:

\begin{align}\label{Term3}
    &(III)\leq \sum_{p=0}^{k-1} \mathbb{E} \left[ \left\lVert H(\theta_i)^{k-p-1}B(\theta_i)L(\theta_i) r^i[p+1] \right\rVert^2 \right] \nonumber \\ & + \sum_{p,q=0,p \ne q}^{k-1} \sqrt{\mathbb{E} \left[ \left\lVert H(\theta_i)^{k-p-1}B(\theta_i)L(\theta_i) r^i[p+1] \right\rVert^2 \right]} \nonumber \\ & \qquad \times \sqrt{\mathbb{E} \left[ \left\lVert H(\theta_i)^{k-q-1}B(\theta_i)L(\theta_i) r^i[q+1] \right\rVert^2 \right]} \nonumber\\ &
    \leq \sum_{p=0}^{k-1} \varrho(\theta_i)\varepsilon(\theta_i)^{k-p-1}\|B(\theta_i)L(\theta_i)\|^2M_{r,i}^2
   \nonumber\\
    &   
   +  \sum_{p,q=0,p \ne q}^{k-1}\varrho(\theta_i)\varepsilon(\theta_i)^{k-p/2-1-q/2}\|B(\theta_i)L(\theta_i)\|^2M_{r,i}^2 \nonumber\\& \leq \frac{\varrho(\theta_i)\|B(\theta_i)L(\theta_i)\|^2M_{r,i}^2}{1-\sqrt{\varepsilon(\theta_i)}} \left[\!\frac{1}{1\!+\!\!\sqrt{\varepsilon(\theta_i)}}\!+\!\frac{1}{1\!-\!\!\sqrt{\varepsilon(\theta_i)}}\!\right]
    \end{align}
where the second equality follows from the fact that $\|\sum_{i=1}^n x_i\|^2 = \sum_{i=1}^n\|x_i\|^2 + \sum_{i,j=1,i \ne j}^n x_ix_j$ and the Cauchy-Schwarz inequality. Further, since $r^i \in \mathcal{X}$, we define $M_{r,i} := \|r^i\|_{\infty}$.

Next, consider the second term in \eqref{Expected_CL_system}. By following similar steps as for the third term, we arrive at the following inequality

\begin{align}\label{Term2}
     & (II) \leq \sum_{p=0}^{k-1} \varrho(\theta_i)\varepsilon(\theta_i)^{k-p-1}\|B(\theta_i)\Pi(\theta_i)\|^2 \mathbb{E}\left[\|e_p\|^2 \right] \nonumber \\ 
    & + \!\!\! \sum_{\substack{p,q=0,\\ p \ne q}}^{k-1} \!\!\!\varrho(\theta_i)\varepsilon(\theta_i)^{k\!-\frac{p}{2}\!-\!1\!-\frac{q}{2}}\|B(\theta_i)\Pi(\theta_i)\|^2\!\! \sqrt{\!\mathbb{E}\left[\|e_p\|^2 \right]\!\mathbb{E}\left[\|e_q\|^2 \right]} \nonumber \\
    & \leq \beta(\theta_i)\varrho(\theta_i)\|B(\theta_i)\Pi(\theta_i)\|^2\left[1+ \left(\frac{1-\epsilon(\theta_i)^{k/2}}{1-\epsilon(\theta_i)^{1/2}}\right)^2\right],
\end{align}
where the last inequality follows by a similar reasoning as in  \eqref{temp_costLQ} and $\beta(\theta_i) = \sum_{r=1}^{\tau_u}tr(A(\theta_i)^{{r-1}^\top}A(\theta_i)^{r-1}K_W(\theta_i)) + \frac{1}{1-\|A(\theta_i)\|_F^2(1-\alpha)}- \frac{1}{\alpha}$.
By summing up the terms in \eqref{Term1}, \eqref{Term4}, \eqref{Term3} and \eqref{Term2}, we have 
\begin{align*}
     \limsup_{T \rightarrow \infty} &\frac{1}{T}\mathbb{E}\left[\sum_{k=0}^{T-1}\|X^{*,i}[k]\|^2  \right] \leq 4\varrho(\theta_i) tr(\Sigma_x + \nu_{\theta_i,0}\nu_{\theta_i,0}^\top) \\ & +   4\beta(\theta_i)\varrho(\theta_i)\|B(\theta_i)\Pi(\theta_i)\|^2\left[1+ \frac{ 1}{(1-\sqrt{\epsilon(\theta_i)})^2}\right]\\
     &+ \frac{8\varrho(\theta_i)\|B(\theta_i)L(\theta_i)\|^2M_{r,i}^2}{(1-\sqrt{\varepsilon(\theta_i)})(1-\varepsilon(\theta_i))} + \frac{4\varrho(\theta_i)tr(K_0)}{1-\varepsilon(\theta_i)}.
\end{align*}
Finally, since $\Theta$ is a finite set, we have the desired result.
\end{proof}

Next, we prove that the equilibrium policy obtained above constitutes an $\epsilon$--Nash equilibrium for the $N$--agent system \eqref{system} with the cost \eqref{LQT}.

\subsection{$\epsilon$--Nash Analysis}
In this subsection, we show that the decentralized equilibrium control policy obtained from the MF analysis is approximately Nash for the finite-agent system.
We start by stating the following Proposition, which shows that the equilibrium MF trajectory approximates the finite-agent state average in the mean-square sense.

\begin{proposition}\label{Approx_behav}
Suppose that Assumptions \ref{As1}-\ref{As4} hold. Then, the equilibrium MF trajectory converges (in the mean-square sense) to the coupling term in \eqref{LQT} with a rate of $\Os \big( \frac{1}{ \min_{\theta \in \Theta} |N_\theta|} \big)$, where $N_{\theta} \subset [N]$ is the subset of agents of type $\theta$, given its cardinality $|N_\theta| > 0, \forall \theta \in \Theta$. More precisely, we have that $\epsilon_T(N) = \Os \big( \frac{1}{ \min_{\theta \in \Theta} | N_\theta |} \big)$, where 
$\epsilon_T(N) = \frac{1}{T} \sum_{k=0}^{T-1}\EE \big[ \big\lVert\frac{1}{N}\sum_{j=1}^N X^{*,j}[k] - \bar{X}^*[k] \big\rVert^2\big].$
\end{proposition}

\begin{proof}
Consider the dynamics process of an agent $i$ of type $\theta$, $X^*_{\theta,i}$ and the mean-field of type $\theta$, $\hat{X}^*_\theta = \EE[X^*_{\theta,i}]$ under the MFE as:
\begin{align*}
    X^*_{\theta,i}[k+1] = & (A(\theta) - B(\theta) \Pi(\theta)) X^*_{\theta,i}[k] \\
    & - B(\theta) K^*_1(\theta) \bar{X}^*[k] \!+ K^*_2(\theta) e_{\theta,i}[k] + W_{\theta,i}[k], \\
    \hat{X}^*_\theta[k+1] \!= & (A(\theta)\! -\! B(\theta) \Pi(\theta))\hat{X}^*_\theta[k] \!- \!B(\theta) K^*_1(\theta) \bar{X}^*[k].
\end{align*}
Here, the gain $K_1^*(\theta)$ can be computed from \eqref{bckw_dyn}, \eqref{clsys} and the linearity of mean-field (Theorem \ref{Th:existUniq}) using Proposition 1 in \cite{zaman2022reinforcement} and $K_2^*(\theta) =  - B(\theta) \Pi(\theta)$. 
We know that under boundedness of $A(\theta)$ (Assumption \ref{AsAbound}), $\EE[e_{\theta,i}[k]e_{\theta,i}[k]^\top]$ is also bounded. Using Theorem \ref{Th:existUniq} we also recall $\bar{X}^*[k] = \sum_{\theta \in \Theta} \mathbb{P}(\theta) \hat{X}^*_\theta[k]$. Now let us define the empirical mean-field of type $\theta$, $Y^*_\theta[k]$ as;
\begin{align*}
    Y^*_\theta[k] := \frac{1}{|N_\theta|} \sum_{i \in N_\theta} X^*_{\theta,i}[k].
\end{align*}
Note that $N = \sum_{\theta \in \Theta} |N_\theta|$. Then, the dynamics of the above process is given as
\begin{align*}
    Y^*_\theta[k+1] = &(A(\theta)-B(\theta) \Pi(\theta)) Y^*_\theta[k]\\
    & - B(\theta) K^*_1(\theta) \bar{X}^*[k] + K^*_2 e_\theta[k] + W_\theta[k],
\end{align*}
where $e_\theta[k] = \frac{1}{|N_\theta|} \sum_{i \in N_\theta} e_{\theta,i}[k]$ and $W_\theta[k] = \frac{1}{|N_\theta|} \sum_{i \in N_\theta} W_{\theta,i}[k]$. Then, it is easy to see that 
\begin{align}
    tr(\EE[e_\theta[k] e^\top_\theta[k]]) & = \Os\bigg(\frac{1}{|N_\theta|}\bigg), \label{eq:cov_e_theta} \\
    tr(\EE[ W_\theta[k] W^\top_\theta[k] ]) & = \Os\bigg(\frac{1}{|N_\theta|}\bigg). \label{eq:cov_W_theta}
\end{align}
due to $e_{\theta,i}$ and $W_{\theta,i}$ being independent among $i$'s and the $\EE[e_{\theta,i}[k]e_{\theta,i}[k]^\top]$ being bounded. Next, let us define the deviation between the true mean-field and the empirical mean-field as $\tilde{Z}^*_\theta[k] := \hat{X}^*_\theta[k] - Y^*_\theta[k].$
The dynamics of this deviation is
\begin{align*}
    \tilde{Z}^*_\theta[k+1] = (A(\theta)-B(\theta) \Pi(\theta)) \tilde{Z}^*_\theta[k] + \tilde{W}_\theta[k],
\end{align*}
where $(A(\theta)-B(\theta) \Pi(\theta))$ is Hurwitz from Proposition \ref{optimalTrackControl} and $\tilde W_\theta[k] = K^*_2(\theta) e_\theta[k] + W_\theta[k]$, such that
\begin{align}\label{CovTemp}
    tr(\EE[\tilde W_\theta[k] \tilde W^\top_\theta[k]]) = \Os\bigg(\frac{1}{|N_\theta|}\bigg),
\end{align}
using \eqref{eq:cov_e_theta}-\eqref{eq:cov_W_theta}. Define $\epsilon^\theta_T(N_\theta)$ as:
\begin{align*}
    \epsilon^\theta_T(N_\theta) := \frac{1}{T} \EE\sum_{k=0}^{T-1} \lVert \tilde{Z}^*_\theta[k] \rVert^2 = \frac{1}{T} \EE\sum_{k=0}^{T-1} \lVert \hat{X}^*_\theta[k] - Y^*_\theta[k] \rVert^2.
\end{align*}
As the Markov chain induced by the stable gain matrix $(A(\theta)-B(\theta) \Pi(\theta))$ is ergodic,
we can deduce
\begin{align*}
    \limsup_{T \rightarrow \infty} \epsilon^\theta_T(N_\theta) = \EE_{\tilde Z_\theta \sim \Ns(0,\Sigma^\infty_\theta)} [\lVert \tilde Z^*_\theta \rVert^2] = \Os\bigg(\frac{1}{|N_\theta|}\bigg).
\end{align*}
This is due to $\EE[\tilde Z^*_\theta[k]] \xrightarrow{k \rightarrow \infty} 0$, \eqref{CovTemp}, and the steady-state covariance matrix of the Markov chain (induced by gain matrix $(A(\theta)-B(\theta) \Pi(\theta))$), denoted by $\Sigma^\infty_\theta$, being the solution to the Lyapunov equation,
\begin{align*}
    \Sigma^\infty_\theta = & \EE[\tilde W_\theta[k] \tilde W^\top_\theta[k]] \\
    & + (A(\theta)-B(\theta) \Pi(\theta)) \Sigma^\infty_\theta (A(\theta)-B(\theta) \Pi(\theta))^\top,
\end{align*}
where $tr(\EE[\tilde W_\theta[k] \tilde W^\top_\theta[k]]) = \Os\big(1/|N_\theta|\big)$ and thus $ \lVert \Sigma^\infty_\theta \rVert = \Os\big(1/|N_\theta|\big)$.
Finally, since
\begin{align*}
    & \epsilon_T(N) = \frac{1}{T} \sum_{k=0}^{T-1}\EE \bigg[ \bigg\lVert\frac{1}{N}\sum_{j=1}^N X^{*}_{\theta,j}[k] - \bar{X}^*[k] \bigg\rVert^2\bigg]\\
    & = \frac{1}{T} \sum_{k=0}^{T-1}\EE \bigg[ \bigg\lVert \sum_{\theta \in \Theta} 
    \sum_{j \in N_\theta}  \frac{X^{*}_{\theta,j}[k]}{N} - \sum_{\theta \in \Theta} \mathbb{P}(\theta) \hat{X}_\theta^*[k] \bigg\rVert^2\bigg] \\
    & = \frac{1}{T} \sum_{k=0}^{T-1}\EE \bigg[ \bigg\lVert \sum_{\theta \in \Theta} \frac{|N_\theta|}{N} Y^*_\theta[k]  - \sum_{\theta \in \Theta} \mathbb{P}(\theta) \hat{X}_\theta^*[k] \bigg\rVert^2\bigg] \\
    & = \frac{1}{T} \sum_{k=0}^{T-1}\EE \bigg[ \bigg\lVert \sum_{\theta \in \Theta} \big( \mathbb{P}_N(\theta) Y^*_\theta[k] -  \mathbb{P}(\theta) \hat{X}_\theta^*[k] \big) \bigg\rVert^2\bigg] \\
    & \leq c^1_\theta \sum_{\theta \in \Theta} \frac{1}{T} \sum_{k=0}^{T-1}\EE \big[ \big\lVert  \mathbb{P}_N(\theta) Y^*_\theta[k] -  \mathbb{P}(\theta) \hat{X}_\theta^*[k]  \big\rVert^2\big] \\
    & \leq c^2_\theta \sum_{\theta \in \Theta} \mathbb{P}(\theta) \frac{1}{T} \sum_{k=0}^{T-1}\EE \big[ \big\lVert Y^*_\theta[k] -  \hat{X}_\theta^*[k]  \big\rVert^2\big] \\
    & \hspace{2.5cm} + c^3_\theta |\mathbb{P}_N(\theta) - \mathbb{P}(\theta)| \frac{1}{T} \sum_{k=0}^{T-1} \EE\big[ \big\lVert Y^*_\theta[k] \rVert^2 \big]
            \end{align*}
    \begin{align*}
    & = \sum_{\theta \in \Theta} \mathbb{P}(\theta) \epsilon^\theta_T(N_\theta) + \Os(\lvert \mathbb{P}_N(\theta) - \mathbb{P}(\theta)\rvert),
\end{align*}
where $c^1_\theta,c^2_\theta,c^3_\theta \geq 0$ are finite constants. The last step is obtained using the definition of $\epsilon^\theta_T(N_\theta)$ and the fact that $\frac{1}{T} \sum_{k=0}^{T-1} \EE\big[ \big\lVert Y^*_\theta[k] \rVert^2 \big]$ is bounded. Since $\epsilon^\theta_T(N_\theta) = \Os \big( 1/ |N_\theta| \big)$ and $|\mathbb{P}_N(\theta) - \mathbb{P}(\theta)| = \Os \big( 1/N \big)$, we can deduce $ \epsilon_T(N) = \Os \big( \frac{1}{ \min_{\theta \in \Theta} |N_\theta|} \big)$.
\end{proof}
Next, we prove the $\epsilon$--Nash property of the equilibrium control policy. To this end, we introduce the following metrics. Also, we supress the arguments of $Q$ and $R$ for brevity.
Let 
\begin{align}\label{CostepsNash1}
&J_i^N(\mu^{*,i}, \mu^{*,-i})= \\& \limsup_{T\rightarrow \infty}\frac{1}{T}\mathbb{E}\left[ \sum_{k=0}^{T-1}\Bigg\|X^{*,i}[k]-\frac{1}{N}\sum_{j=1}^{N}{X^{*,j}[k]}\Bigg\|^2_{Q}\!\!\! +\! \|U^{*,i}[k]\|^2_{R}\right] \nonumber \end{align}
be the cost of agent $i$ when all agents follow the policy under the MFE;
\begin{align}\label{CostepsNash2}
&J(\mu^{*,i}, \bar{X}^*)= \\& \limsup_{T\rightarrow \infty}\frac{1}{T}\mathbb{E}\left[ \sum_{k=0}^{T-1}\|X^{*,i}[k] -\bar{X}^*[k]\|^2_{Q} + \|U^{*,i}[k]\|^2_{R}\right]\nonumber
\end{align}
be the cost of agent $i$ with the consensus term replaced by the equilibrium MF trajectory; and
\begin{align}
&J_i^N(\pi_c^i, \mu^{*,-i})= \\& \limsup_{T\rightarrow \infty}\frac{1}{T}\mathbb{E}\!\left[ \sum_{k=0}^{T-1}\Bigg\|X^{i,\pi_c^i}[k]-\frac{1}{N}\sum_{j=1}^{N}{X^{*,j}[k]}\Bigg\|^2_{Q}\!\!\! + \! \|V^{i}[k]\|^2_{R}\right]\! \nonumber
\end{align}
be the cost of agent $i$ when it deviates from the MF policy while all other agents follow the MF policy.
Here, $X^{i,\pi_c^i}[k]$ is the state of agent $i$ at time $k$ when it chooses a control policy $\pi_c^i \in {\mathcal{M}}_i^{c,con}$. Furthermore, the control action $V^i[k]$ is derived from $\pi^i_c$. 
To make the argument and the result more concrete, we now have the following definition:
\begin{definition}\cite{aggarwal2022linear}\label{epsNashDefn}
The set of control policies $\{\mu^{*,i}, 1 \leq i \leq N\}$ constitute an $\epsilon$--Nash equilibrium with respect to the cost functions $\{J^N_i, 1 \leq i \leq N\}$, if there exists $\epsilon > 0$, such that
\begin{align*}
J^N_i(\mu^{*,i}, \mu^{*,-i}) \leq \inf_{\pi^i \in \mathcal{M}^{c,con}_i} J^N_i(\pi_c^{i}, \mu^{*,-i}) + \epsilon,~~\forall i \in [N].
\end{align*}
\end{definition}
Next, we give the following theorem stating the $\epsilon$-Nash result, i.e., the control laws prescribed by the MFE constitute an $\epsilon$-Nash equilibrium in the finite population case.

\begin{theorem}\label{Th:eps_Nash}
Suppose Assumptions \ref{As1}-\ref{As4} hold. Then the set of decentralized control policies $\{\mu^{*,i}, 1 \leq i \leq N\}$, constitute an $\epsilon$--Nash equilibrium for the $N$--agent LQ-mean field game with bandwidth limits. In particular, we have that
\begin{align}\label{epsNash}
J^N_i\!(\mu^{*,i}, \mu^{*,-i}) \leq \!\!\!& \inf_{\pi_c^i \in \mathcal{M}^{c,con}_i} J^N_i(\pi_c^{i}, \mu^{*,-i})\! + \!\Os\!\! \left(\! \frac{1}{ \sqrt{\min\limits_{\theta \in \Theta} |N_\theta|}}\!\right), 
\end{align}
i.e., $\epsilon = \Os (\! \frac{1}{ \sqrt{\min_{\theta \in \Theta} |N_\theta|}})$ in Definition \ref{epsNashDefn}.
\end{theorem}

\begin{proof}
We prove the theorem in two steps. In the first step, we derive an upper bound on $J_i^N(\mu^{*,i}, \mu^{*,-i})- J(\mu^i, \bar{X}^*)$, and in step 2, on $J(\mu^i, \bar{X}^*) - J_i^N(\pi_c^i, \mu^{*,-i})$. Finally, we combine these bounds to arrive at \eqref{epsNash}.

\noindent \textbf{Step 1:} Consider the following:
\begin{align}
&J_i^N(\mu^{*,i}, \mu^{*,-i})- J(\mu^i, \bar{X}^*) \leq J_i^N(\mu^{*,i}, \mu^{*,-i})- J(\mu^{*,i}, \bar{X}^*) \nonumber \\
& = \limsup_{T\rightarrow \infty}\frac{1}{T}\mathbb{E}\left[ \sum_{k=0}^{T-1}\Bigg\|\bar{X}^*[k]-\frac{1}{N}\sum_{j=1}^{N}{X^{*,j}[k]}\Bigg\|^2_{Q}\right. \nonumber \\ &\left. ~~~~+ 2 \left( X^{*,i}[k]-\bar{X}^*[k]\right)^\top Q\left( \bar{X}^*[k]-\frac{1}{N}\sum_{j=1}^{N}{X^{*,j}[k]}\right) \right] \nonumber\\
& \leq \limsup_{T\rightarrow \infty}\|Q\|\epsilon_T(N) \nonumber \\& + \limsup_{T\rightarrow \infty}2\|Q\|\sqrt{\epsilon_T(N)\mathbb{E}\left[ \frac{1}{T}\sum_{k=0}^{T-1}\|X^{*,i}[k]-\bar{X}^*[k]\|^2 \right]} \label{X3}
\end{align}
where \eqref{X3} follows from the Cauchy-Schwarz inequality.

\noindent \textbf{Step 2: }Consider the following:
\begin{align*}
&J_i^N(\pi_c^i, \mu^{*,-i})= \\ & \limsup_{T\rightarrow \infty}\frac{1}{T}\mathbb{E}\left[ \sum_{k=0}^{T-1}\Bigg\|X^{i,\pi_c^i}[k]-\frac{1}{N}\sum_{j=1}^{N}{X^{*,j}[k]}\Bigg\|^2_{Q} \!\! + \|V^{i}[k]\|^2_{R}\right] \\
& = J(\pi_c^i, \bar{X}^*)\! - \! \limsup_{T\rightarrow \infty}\frac{1}{T}\mathbb{E}\! \left[ \sum_{k=0}^{T-1} \Bigg\|\frac{1}{N}\sum_{j=1}^{N}{X^{*,j}[k]}\! -\! \bar{X}^{*,i}[k]\Bigg\|^2_Q \right. \\ & \left. + 2 \left(\bar{X}^*[k] - X^{i,\pi_c^i}[k] \right)^\top Q \left( \bar{X}^*[k] -\frac{1}{N}\sum_{j=1}^{N}{X^{*,j}[k]}\right) \right] \\
& \geq J(\mu^i, \bar{X}^*)\! -\! \limsup_{T\rightarrow \infty}\frac{1}{T}\mathbb{E}\! \left[ \sum_{k=0}^{T-1} \Bigg\|\frac{1}{N}\sum_{j=1}^{N}{X^{*,j}[k]}\! - \! \bar{X}^{*,i}[k]\Bigg\|^2_Q \right. \\ & \left. + 2 \left(\bar{X}^*[k] - X^{i,\pi_c^i}[k] \right)^\top Q \left( \bar{X}^*[k] -\frac{1}{N}\sum_{j=1}^{N}{X^{*,j,\pi_c^j}[k]}\right) \right].
\end{align*} 
Finally, using Cauchy-Schwarz inequality in the same manner as for \eqref{X3}, we get
\begin{align*}
&J(\mu^i, \bar{X}^*) - J_i^N(\pi_c^i, \mu^{*,-i}) \leq \limsup_{T\rightarrow \infty}\|Q\|\epsilon_T(N) \nonumber \\& + \limsup_{T\rightarrow \infty}2\|Q\|\sqrt{\epsilon_T(N)\mathbb{E}\left[ \frac{1}{T}\sum_{k=0}^{T-1}\|X^{i,\pi_c^i}[k]-\bar{X}^*[k]\|^2 \right]}. 
\end{align*}

By using Lemma \ref{Th:CLanalysis}, we have that there exist $M_2,~T_2 >0$, such that $\frac{1}{T}\mathbb{E}\left[\sum_{k=0}^{T-1}\|X^{*,i}[k]\|^2 \right] < M_2,~ \forall T > T_2$. Further, since $\bar{X}^* \in \mathcal{X}$, there exist $M_3, ~T_3>0$, such that $\frac{1}{T}\mathbb{E} \left[ \sum_{k=0}^{T-1}\|\bar{X}^{*,i}[k]\|^2 \right] < M_3, ~\forall T>T_3$. Finally, since $\mathcal{M}^{d,con}_i \subseteq \mathcal{M}_i^{c,con}$,  $\inf_{\pi_c^i \in \mathcal{M}_i^{c,con}}J_i^N(\pi_c^i, \mu^{*,-i}) \leq J_i^N(\mu^{*,i}, \mu^{*,-i})$, we may consider $\pi_c^i \in \mathcal{M}_i^{c,con}$ such that there exist $M_4,~T_4 >0$, with the property that $\frac{1}{T}\mathbb{E}\left[ \sum_{k=0}^{T-1}\|X^{i, \pi_c^i}[k]\|^2 \right] < M_4,~\forall T>T_4$. Choose $T_5 = \max\{T_1,T_2,T_3,T_4\}$, and let $T > T_5$. Then, we have that

\begin{align*}
    & J^N_i(\mu^{*,i}, \mu^{*,-i}) \leq \!\!\! \inf_{\pi_c^i \in \mathcal{M}^c_i}\! J^N_i(\pi_c^{i}, \mu^{*,-i})\! + \! 2\limsup_{T\rightarrow \infty}\|Q\|\epsilon_T(N) \nonumber \\ & + \limsup_{T\rightarrow \infty}2\sqrt{2}\|Q\| \times \nonumber \\ & \qquad \sqrt{\frac{\epsilon_T(N)}{T}\left(\mathbb{E}\left[ \sum_{k=0}^{T-1}\|X^{*,i}[k]\|^2 \right] + \mathbb{E}\left[ \sum_{k=0}^{T-1}\|\bar{X}^*[k]\|^2 \right]\right)}\nonumber \\ & + \limsup_{T\rightarrow \infty}2\sqrt{2}\|Q\| \times \nonumber \\ & ~~ \sqrt{\frac{\epsilon_T(N)}{T}\left(\mathbb{E}\left[ \sum_{k=0}^{T-1}\|X^{i,\pi_c^i}[k]\|^2 \right] + \mathbb{E}\left[ \sum_{k=0}^{T-1}\|\bar{X}^*[k]\|^2 \right]\right)}.
\end{align*}
By defining $\epsilon:= \mathcal{O}(\limsup_{T\rightarrow \infty} \sqrt{\epsilon_T(N)})$, and using Proposition \ref{Approx_behav}, we get \eqref{epsNash}. The proof is thus complete.
\end{proof}

\section{Simulations} \label{sec:sims}
In this section, we provide an empirical analysis of the theoretical results. We simulate an $N$-agent system with scalar dynamics. We first point out a key observation which may seem not obvious. Consider the set of $A$ and $K_W$ values for $7$ agents as $A \!\!=\! \!\{0.1,0.3,0.5,0.7,1.0,1.3,1.4,1.5\}$ and $ K_W = \{3.0$, $5.0, 1.0, 2.0, 4.0, 0.1, 2.0\}$. It may be tempting to think that the values for the $\tau^i$'s would be lower for the unstable agents than for the stable ones.  
But the values of $\tau^i$ obtained were $\tau = \{1,0,1,1,1,2,1\}$ for $R_d = 4$. This shows, for instance, that agent 6 is not connected over the network for a longer period than a more stable agent 5. But this should be expected because $\mathbb{E}[\|e^i[k]\|^2]$ for all $i$ depends not only on $A_i$ but also on $K_{W^i}$. Thus, a higher $\tau^6$ is explained by a low value of $K_{W^6}$.

Next, Fig. \ref{TauVsRd} shows a variation of $\tau^i$ and $R_d$ for $N=6$.
The agent parameters used are $A = \{0.1,0
299,0.498,0.697, $ $0.896,1.095\}$ and $K_W = \{0.3, 0.9,$ $ 1.5,2.5, 4, 4.5\}$.
It can be seen that with more restricted bandwidth constraints, the agents are not connected to their respective controllers for longer intervals. Further, it can also be seen from the figure that for small values of $R_d$, stability of a system becomes more important, while for higher $R_d$ values, the noise covariance becomes a dominating factor. This observation is aligned with intuition and can also be inferred from the equations for the error covariance and the cost function.

\begin{figure}[h!]
	\centering
	\includegraphics[width=\linewidth]{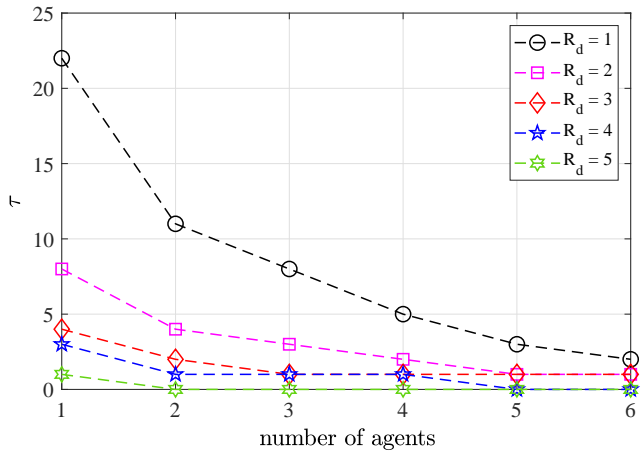}
	\vspace{-0.1cm}
	\caption{\small{The plot shows the variation of $\tau$ with $R_d$ for $N=6$. 
	}}
	\label{TauVsRd}
\end{figure}

Next, we empirically demonstrate the asymptotic optimality of the hard-bandwidth policy as proved in Theorem \ref{Th:Asymptotic_optimality}. We simulate the behavior of WAoI under: 1) the randomized scheduling policy (Section \ref{subsec:rand_sched_pol}), and 2) the bandwidth constrained scheduling policy (Section \ref{subsec:hard_cons_pol}) with a bandwidth of $R_d = 0.6 N$. We plot the average WAoI of both these policies from $N =5$ till $N =1500$, 
and a simulation time of $T= 100000$ seconds. Fig. \ref{Fig:Asympt_opt} shows that the WAoI for both asymptotically approach each other as $N$ increases. This means that as the number of agents increases, the approximation error between the bandwidth constrained and the soft-constrained problems decreases to zero.

\begin{figure}[h!]
     \centering
 	\includegraphics[width=\linewidth]{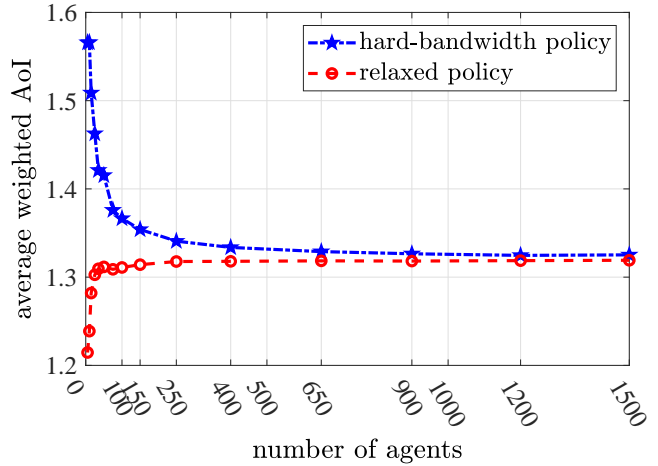}
 	\caption{The plot shows the performance comparison of the optimal policy for the Soft-constrained Problem \ref{Problem4} and Bandwidth-constrained Problem \ref{Problem2} for $R_d = 0.6 N$.}
 	\label{Fig:Asympt_opt}
  	\vspace{-0.5cm}
 \end{figure}
 Finally, we empirically evaluate the behavior of $N=800$ agents under the MFE policy \eqref{optimal_control}. To this end, we plot the average cost \eqref{CostepsNash1} of an agent as a function of the proportion of available downlink bandwidth in Fig. \ref{Fig:Avg_cost_Vs_Bandwidth_ratio} for $A = 1.15$, $K_W = 2$, $T= 500$ seconds, and a single type. We note that the plotted cost is also equivalent to the average cost over all agents since the multi-agent system is homogeneous. The
figure shows a box plot depicting the median (red line) and spread (box) of the average cost per agent over 100 runs for each value of $\alpha$. From the figure, we see that the average cost decreases as the available bandwidth increases, aligned with intuition.
 \begin{figure}[h!]
     \centering
 	\includegraphics[width=\linewidth]{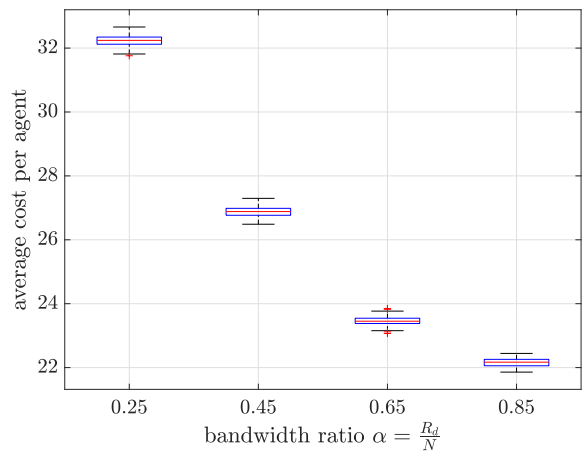}
 	\vspace{-0.35cm}
 	\caption{Average cost per agent when bandwidth ratio is equal to $\alpha = \frac{R_d}{N} = 0.25, 0.45, 0.65$ and $0.85$.}
 	\label{Fig:Avg_cost_Vs_Bandwidth_ratio}
 	\vspace{-0.5cm}
 \end{figure}

 

\section{Conclusion}\label{sec:conclusion}
In this paper, we have studied an $N+1$ player game problem in which $N$ agents aim to achieve consensus while the BS schedules information over a bandwidth-constrained network to these agents using a WAoI metric. To solve the scheduling problem, we have first formulated a relaxed version of the problem and constructed a stationary randomized optimal scheduling policy for it. We have then provided a scheduling policy for the hard bandwidth constrained problem inspired by the optimal policy of the relaxed problem. Finally, we have shown that the former is asymptotically optimal to the latter as $N \rightarrow \infty$. Next, we have solved the game problem between the $N$--agents using the mean-field game approach and employing the obtained scheduling policy. By considering a limiting system as $N \rightarrow \infty$, we have first proved the existence and uniqueness of the mean-field equilibrium and then have shown the $\epsilon$--Nash property of the equilibrium solution for the finite-agent system. Finally, we have validated the theoretical results using empirical simulations for both the proposed scheduling and the control policies. Future research directions involve extending the formulation to include erasure channels (with probabilistic transmission from scheduler to the decoders) and channels with random transmission delays.

\bibliographystyle{IEEEtran}
\bibliography{references}

\end{document}